\documentclass[11pt]{article}

\usepackage{amsmath}
\usepackage{amssymb}
\usepackage{amsthm}
\usepackage{graphics}
\usepackage[latin1]{inputenc}
\usepackage[english]{babel}
\usepackage[T1]{fontenc}

\usepackage{enumerate}

\newtheorem{thm}{Theorem}[section]
\newtheorem{prop}{Proposition}[section]
\newtheorem{coro}{Corollary}[section]
\newtheorem{lemma}{Lemma}[section]
\newtheorem{rem}{Remark}[section]

\newtheorem{defi}{Definition}[section]

\newcommand{\R}{\mathbb{R}}             
\newcommand{\N}{\mathbb{N}}             
\newcommand{\Z}{\mathbb{Z}}             
\newcommand{\C}{\mathbb{C}}             
\renewcommand{\H}{\mathcal{H}}          

\newcommand{\D}{\mathcal{D}}            
\renewcommand{\S}{\mathcal{S}}          

\newcommand{\half}{\frac{1}{2}}

\newcommand{\ds}{\displaystyle}

\newcommand{\Section}[1]{\section{#1} \setcounter{equation}{0}}

\begin{document}

\title{Non-uniqueness results for the anisotropic Calderon problem with data measured on disjoint sets}
\author{Thierry Daud\'e \footnote{Research supported by the French National Research Projects AARG, No. ANR-12-BS01-012-01, and Iproblems, No. ANR-13-JS01-0006} $^{\,1}$, Niky Kamran \footnote{Research supported by NSERC grant RGPIN 105490-2011} $^{\,2}$ and Francois Nicoleau \footnote{Research supported by the French National Research Project NOSEVOL, No. ANR- 2011 BS0101901} $^{\,3}$\\[12pt]
 $^1$  \small D\'epartement de Math\'ematiques. UMR CNRS 8088. \\
\small Universit\'e de Cergy-Pontoise \\
\small 95302 Cergy-Pontoise, France  \\
\small thierry.daude@u-cergy.fr\\
$^2$ \small Department of Mathematics and Statistics\\
\small  McGill
University\\ \small Montreal, QC, H3A 2K6, Canada\\
\small nkamran@math.mcgill.ca\\
$^3$  \small  Laboratoire de Math\'ematiques Jean Leray, UMR CNRS 6629 \\
      \small  2 Rue de la Houssini\`ere BP 92208 \\
      \small  F-44322 Nantes Cedex 03 \\
\small francois.nicoleau@math.univ-nantes.fr \\}





\maketitle


\begin{abstract}

In this paper, we give some simple counterexamples to uniqueness for the Calderon problem on Riemannian manifolds with boundary when the Dirichlet and Neumann data are measured on disjoint sets of the boundary. We provide counterexamples in the case of two and three dimensional Riemannian manifolds with boundary having the topology of circular cylinders in dimension two and toric cylinders in dimension three. The construction could be easily extended to higher dimensional Riemannian manifolds.


\vspace{1cm}




\end{abstract}

\tableofcontents
\newpage


\Section{Introduction}

This paper is devoted to the study of various anisotropic Calderon problems on some simple Riemannian manifolds to be described below. Let us first recall some basic facts about the anisotropic Calderon problem in this setting. We refer for instance to \cite{DSFKSU, DSFKLS, GSB, GT1, KS1, LaTU, LaU, LeU} for important contributions to the subject and to the surveys \cite{GT2, KS2, Sa, U1} for the current state of the art.

Let $(M, g)$ be an $n$ dimensional smooth compact Riemannian manifold with smooth boundary $\partial M$. Let us denote by $\Delta_{LB}$ the positive Laplace-Beltrami operator on $(M,g)$. In a local coordinate system $(x^i)_{i = 1,\dots,n}$, the Laplace-Beltrami operator $\Delta_{LB}$ has the expression
$$
\Delta_{LB}=  -\Delta_g = -\frac{1}{\sqrt{|g|}} \partial_i \left( \sqrt{|g|} g^{ij} \partial_j \right),
$$
where $\left(g^{ij}\right)$ is the matrix inverse of the metric tensor $\left(g_{ij}\right)$, where  $|g| = \det \left(g_{ij}\right)$ is the determinant of $g$ and where we use the Einstein summation convention. We recall that the Laplace-Beltrami operator $-\Delta_g$ with Dirichlet boundary conditions is selfadjoint on $L^2(M, dVol_g)$ and has pure point spectrum $\{ \lambda_j^2\}_{j \geq 1}$ so that $0 < \lambda_1^2 \leq \dots \leq \lambda_j^2 \to +\infty$ (see for instance \cite{KKL}).

We consider the Dirichlet problem at a frequency $\lambda^2 \in \R$ on $(M,g)$ such that $\lambda^2 \notin \{ \lambda_j^2\}_{j \geq 1}$. We are interested in the solutions $u$ of
\begin{equation} \label{Eq00}
  \left\{ \begin{array}{cc} -\Delta_g u = \lambda^2 u, & \textrm{on} \ M, \\ u = \psi, & \textrm{on} \ \partial M. \end{array} \right.
\end{equation}
It is well known (see for instance \cite{Sa}) that for any $\psi \in H^{1/2}(\partial M)$, there exists a unique weak solution $u \in H^1(M)$ of (\ref{Eq00}). This allows us to define the Dirichlet-to-Neumann (DN) map as the operator $\Lambda_{g}(\lambda^2)$ from $H^{1/2}(\partial M)$ to $H^{-1/2}(\partial M)$ defined for all $\psi \in \H^{1/2}(\partial M)$ by
\begin{equation} \label{DN-Abstract}
  \Lambda_{g}(\lambda^2) (\psi) = \left( \partial_\nu u \right)_{|\partial M},
\end{equation}
where $u$ is the unique solution of (\ref{Eq00}) and $\left( \partial_\nu u \right)_{|\partial M}$ is its normal derivative with respect to the unit outer normal vector $\nu$ on $\partial M$. Here $\left( \partial_\nu u \right)_{|\partial M}$ is interpreted in the weak sense as an element of $H^{-1/2}(\partial M)$ by
$$
  \left\langle \Lambda_{g}(\lambda^2) \psi | \phi \right \rangle = \int_M \langle du, dv \rangle_g dVol_g,
$$
for any $\psi \in H^{1/2}(\partial M)$ and $\phi \in H^{1/2}(\partial M)$ such that $u$ is the unique solution of (\ref{Eq00}) and $v$ is any element of $H^1(M)$ so that $v_{|\partial M} = \phi$. If $\psi$ is sufficiently smooth, we can check that
$$
  \Lambda_{g}(\lambda^2) \psi = g(\nu, \nabla u)_{|\partial M} = du(\nu)_{|\partial M} = \nu(u)_{|\partial M},
$$
where $\nu$ represents the unit outer normal vector to $\partial M$. Clearly, in that case, an expression in local coordinates for the normal derivative is thus
\begin{equation} \label{DN-Coord}
\partial_\nu u = \nu^i \partial_i u.
\end{equation}

We shall be interested in fact in the \emph{partial} DN maps defined as follows. Let $\Gamma_D$ and $\Gamma_N$ be two open subsets of $\partial M$. We then define the partial DN map $\Lambda_{g,\Gamma_D,\Gamma_N}(\lambda^2)$ as the restriction of the global DN map $\Lambda_g(\lambda^2)$ to Dirichlet data given on $\Gamma_D$ and Neumann data measured on $\Gamma_N$. Precisely, consider the Dirichlet problem
\begin{equation} \label{Eq0}
  \left\{ \begin{array}{cc} -\Delta_g u = \lambda^2 u, & \textrm{on} \ M, \\ u = \psi, & \textrm{on} \ \Gamma_D, \\ u = 0, & \textrm{on} \ \partial M \setminus \Gamma_D. \end{array} \right.
\end{equation}
We thus define $\Lambda_{g,\Gamma_D,\Gamma_N}(\lambda^2)$ as the operator acting on the functions $\psi \in H^{1/2}(\partial M)$ with $\textrm{supp}\,\psi \subset \Gamma_D$ by
\begin{equation} \label{Partial-DNmap}
  \Lambda_{g,\Gamma_D,\Gamma_N}(\lambda^2) (\psi) = \left( \partial_\nu u \right)_{|\Gamma_N},
\end{equation}
where $u$ is the unique solution of (\ref{Eq0}).

The anisotropic partial Calderon problem can be initially stated as: \emph{does the knowledge of the partial DN map $\Lambda_{g,\Gamma_D, \Gamma_N}(\lambda^2)$ at a frequency $\lambda^2$ determine uniquely the metric $g$}? One can think of three subcases of the above problem:
\begin{itemize}
\item \textbf{Full data}: $\Gamma_D = \Gamma_N = \partial M$. In that case, we simply denote by $\Lambda_g(\lambda^2)$ the DN map.

\item \textbf{Local data}: $\Gamma_D = \Gamma_N = \Gamma$, where $\Gamma$ can be any nonempty open subset of $\partial M$. In that case, we denote by $\Lambda_{g, \Gamma}(\lambda^2)$ the DN map.

\item \textbf{Data on disjoint sets}: $\Gamma_D$ and $\Gamma_N$ are disjoint open sets of $\partial M$.
\end{itemize}

Due to a number of gauge invariances, the answer to the above questions is no. Indeed, it is clear from the definition (\ref{Eq0}) - (\ref{Partial-DNmap}) that in any dimension, the partial DN map $\Lambda_{g, \Gamma_D, \Gamma_N}(\lambda^2)$ is invariant under pullback of the metric by the diffeomorphisms of $M$ that are the identity on $\Gamma_D \cup \Gamma_N$, \textit{i.e.}
\begin{equation} \label{Inv-Diff}
  \forall \phi \in \textrm{Diff}(M) \ \textrm{such that} \ \phi_{|\Gamma_D \cup \Gamma_N} = Id, \quad \Lambda_{\phi^*g, \Gamma_D, \Gamma_N}(\lambda^2) = \Lambda_{g, \Gamma_D, \Gamma_N}(\lambda^2).
\end{equation}

In the two dimensional case and for zero frequency $\lambda^2 = 0$, there is another gauge invariance of the DN map due to the conformal invariance of the Laplace-Beltrami operator. More precisely, recall that in dimension $2$
$$
  \Delta_{cg} = \frac{1}{c} \Delta_g,
$$
for any smooth function $c >0$. Therefore, we have in dimension $2$
\begin{equation} \label{Inv-Conf}
  \forall c \in C^\infty(M) \ \textrm{such that} \ c >0 \ \textrm{and} \ c_{|\Gamma_N} = 1, \quad \Lambda_{c g, \Gamma_D, \Gamma_N}(0) = \Lambda_{g, \Gamma_D, \Gamma_N}(0),
\end{equation}
since the unit outer normal vectors $\nu_{cg}$ and $\nu_g$ coincide on $\Gamma_N$ in that case.

Hence the appropriate question (called the \emph{anisotropic Calderon conjecture}) to adress is the following. \\

\noindent \textbf{(Q1)}: \emph{Let $M$ be a smooth compact manifold with smooth boundary $\partial M$ and let $g_{1},\,g_{2}$ be smooth Riemannian metrics on $M$. Let $\Gamma_D, \Gamma_N$ be any open sets of $\partial M$ and $\lambda^2$ be a fixed frequency that does not belong to $\sigma(-\Delta_g)$. If
$$
  \Lambda_{g_1,\Gamma_D, \Gamma_N}(\lambda^2) = \Lambda_{g_2,\Gamma_D, \Gamma_N}(\lambda^2),
$$
then is it true that
$$
  g_1 = g_2,
$$
up to the invariance (\ref{Inv-Diff}) if $\dim M \geq 3$ and up to the invariances (\ref{Inv-Diff}) - (\ref{Inv-Conf}) if $\dim M = 2$ and $\lambda^2 = 0$}?

In $\dim M \geq 3$, we can adress another relevant (and simpler!) problem by assuming that the Riemannian manifolds $(M,g_1)$ and $(M,g_2)$ belong to the same conformal class, \textit{i.e.} there exists a smooth positive function $c$ (called the conformal factor) such that $g_2 = c g_1$. In that case, $g_1$ is considered as the background known metric and the problem consists in determining the unknown scalar function $c$ from the DN map $\Lambda_{c g_1,\Gamma_D, \Gamma_N}(\lambda^2)$. Precisely, the question becomes: \\

\noindent \textbf{(Q2)}: \emph{Let $(M,g)$ be a smooth compact Riemannian manifold with boundary $\partial M$ and let $\Gamma_D, \Gamma_N$ be any open sets of $\partial M$. Let $c$ be a smooth positive function on $M$. If
$$
  \Lambda_{c g,\Gamma_D, \Gamma_N}(\lambda^2) = \Lambda_{g,\Gamma_D, \Gamma_N}(\lambda^2),
$$
show that there exists a diffeomorphism $\phi: \, M \longrightarrow M$ with $\phi_{| \, \Gamma_D \cup \Gamma_N} = Id$ so that}
\begin{equation} \label{Inv-Conformal}
  \phi^* g = cg.
\end{equation}

Note that in the case of full data $\Gamma_D = \Gamma_N = \partial M$ or more generally in the case where $\Gamma_D \cup \Gamma_N = \partial M$, it is known that any diffeomorphism $\phi: \, M \longrightarrow M$  which satisfies $\phi_{|\partial M} = Id$ and $\phi^* g = cg$ must be the identity \cite{Li}. Therefore, in either of these particular cases, there is no ambiguity arising from diffeomorphisms and (\ref{Inv-Conformal}) should be replaced by the condition
\begin{equation} \label{Inv-Conformal-1}
  c = 1, \quad \textrm{on} \ M.
\end{equation}

A last version of the anisotropic Calderon problem which in some sense generalizes \textbf{(Q2)} is the following. Consider the solution of the Schr\"odinger equation on $(M,g)$ with potential $V \in L^\infty(M)$
\begin{equation} \label{Eq0-Schrodinger}
  \left\{ \begin{array}{cc} (-\Delta_g + V) u = \lambda^2 u, & \textrm{on} \ M, \\ u = \psi, & \textrm{on} \ \Gamma_D, \\ u = 0, & \textrm{on} \ \partial M \setminus \Gamma_D. \end{array} \right.
\end{equation}
It is well known (see again for instance \cite{DSFKSU, Sa}) that if $\lambda^2$ does not belong to the Dirichlet spectrum of $-\Delta_g +V$, then for any $\psi \in H^{1/2}(\partial M)$, there exists a unique weak solution $u \in H^1(M)$ of (\ref{Eq0-Schrodinger}). This allows us to define the partial Dirichlet-to-Neumann map $\Lambda_{g, V, \,\Gamma_D, \Gamma_N}(\lambda^2)$ for all $\psi \in \H^{1/2}(\partial M)$ with supp $\psi \subset \Gamma_D$ by
\begin{equation} \label{DN-Abstract-Schrodinger}
  \Lambda_{g, V,\Gamma_D, \Gamma_N}(\lambda^2) (\psi) = \left( \partial_\nu u \right)_{|\Gamma_N},
\end{equation}
where $u$ is the unique solution of (\ref{Eq0-Schrodinger}) and $\left( \partial_\nu u \right)_{|\Gamma_N}$ is its normal derivative with respect to the unit outer normal vector $\nu$ on $\Gamma_N$. Once again, we assume that $g$ is a known background metric and the problem consists in determining the unknown potential $V \in L^\infty(M)$ from the DN map $\Lambda_{g, V, \,\Gamma_D, \Gamma_N}(\lambda^2)$. Precisely, the question is: \\

\noindent \textbf{(Q3)}: \emph{Let $(M,g)$ be a smooth compact Riemannian manifold with smooth boundary $\partial M$ and let $\Gamma_D, \Gamma_N$ be any open sets of $\partial M$. Let $V_1$ and $V_2$ be potentials in $L^\infty(M)$. If
$$
  \Lambda_{g, V_1, \Gamma_D, \Gamma_N}(\lambda^2) = \Lambda_{g, V_2, \Gamma_D, \Gamma_N}(\lambda^2),
$$
is it true that}
$$
  V_1 = V_2?
$$

There is a straightforward link between \textbf{(Q2)} and \textbf{(Q3)} in the case of zero frequency $\lambda^2 = 0$ and global data. The main point is the observation that the Laplace-Beltrami operator transforms under conformal scalings of the metric by
\begin{equation} \label{ConformalScaling}
  \Delta_{cg} u = c^{-\frac{n+2}{4}} \left( \Delta_g + q \right) \left( c^{\frac{n-2}{4}} u \right),
\end{equation}
where
\begin{equation} \label{q}
  q = c^{-\frac{n-2}{4}} \Delta_{g} c^{\frac{n-2}{4}}.
\end{equation}	
Then we can show that if $c$ is a positive smooth function on $M$ such that
$$
  c_{|\Gamma_D \cup \Gamma_N} = 1, \quad \left( \partial_{\nu} c \right)_{|\Gamma_N} = 0,
$$
then
\begin{equation} \label{Inv-Conf-V}
  \Lambda_{cg, V,\,\Gamma_D,\Gamma_N}(0) = \Lambda_{g, cV-q,\,\Gamma_D,\Gamma_N}(0),
\end{equation}
where $q$ is given by (\ref{q}). The proof of (\ref{Inv-Conf-V}) is an immediate adaptation of the proof of the Prop 8.2 in \cite{DSFKSU}. In particular, we have (using the preceding notations)
\begin{equation} \label{Link-Q23}
  \Lambda_{cg,\,\Gamma_D,\Gamma_N}(0) = \Lambda_{g, -q,\,\Gamma_D,\Gamma_N}(0),
\end{equation}
where $q$ is given by (\ref{q}).

Let us show that \textbf{(Q3)} implies \textbf{(Q2)} in the case of zero frequency and global data (\textit{i.e.} when $\Gamma_D = \Gamma_N = \partial M$). Assume that \textbf{(Q3)} is true and assume that for two metrics $g$ and $cg$, we have
\begin{equation} \label{u1}
  \Lambda_{c g}(0) = \Lambda_{g}(0).
\end{equation}
Then by boundary determination (\cite{DSFKSU, KY, LeU}, we can show that $c_{|\partial M} = 1$ and $\left( \partial_{\nu} c \right)_{|\partial M} = 0$. Hence, we can use (\ref{Link-Q23}) to show that (\ref{u1}) is equivalent to
 \begin{equation} \label{u2}
  \Lambda_{g, -q}(0) = \Lambda_{g, 0}(0),
\end{equation}
with $q$ given by (\ref{q}) and $\Lambda_{g, -q}(0)$ stands for the global DN map. Finally, our hypothesis that \textbf{(Q3)} holds true now asserts that $q = 0$, or in other words that $\Delta_g c^{\frac{n-2}{4}} = 0$. Since $c_{|\partial M} = 1$, uniqueness of solutions for the Dirichlet problem shows that $c = 1$ on $M$ and \textbf{(Q2)} is proved. \\

The most complete results on the anisotropic Calderon problems \textbf{(Q1)}, \textbf{(Q2)} and \textbf{(Q3)} have been obtained in the case of full data ($\Gamma_D = \Gamma_N = \partial M$) and local data ($\Gamma_D = \Gamma_N = \Gamma$ with $\Gamma$ any open subset of $M$) for \emph{vanishing frequency $\lambda^2 = 0$}. In dimension $2$, the anisotropic Calderon problem \textbf{(Q1)} for global and local data with $\lambda^2 = 0$ was shown to be true for connected Riemannian surfaces in \cite{LaU}. We also refer to \cite{ALP} for similar results answering \textbf{(Q1)} for global and local data in the case of anisotropic conductivities which are only $L^\infty$ on bounded domains of $\R^n$.
A positive answer to \textbf{(Q1)} for global and local data and zero frequency $\lambda^2 = 0$ in dimension $3$ or higher has been given for compact connected real analytic Riemannian manifolds with real analytic boundary first in \cite{LeU} under some topological assumptions relaxed later in \cite{LaU, LaTU} and for compact connected Einstein manifolds with boundary in \cite{GSB}. Note that Einstein manifolds are real analytic in their interior. Let us point out here that no connectedness assumption on the  measurement set $\Gamma$ was made in the works \cite{GSB, LaU}.

The general full or local data anisotropic Calderon problem \textbf{(Q1)} in dimension $3$ or higher remains a major open problem. A few deep results concerning the partial questions \textbf{(Q2)} and \textbf{(Q3)} have been obtained recently in \cite{DSFKSU, DSFKLS} for some classes of smooth compact Riemannian manifolds with boundary that are \emph{conformally transversally anisotropic}, \textit{i.e.} Riemannian manifolds $(M,g)$ such that
$$
  M \subset \subset \R \times M_0, \quad g = c ( e \oplus g_0),
$$
where $(M_0,g_0)$ is a $n-1$ dimensional smooth compact Riemannian manifold with boundary, $e$ is the Euclidean metric on the real line, and
$c$ is a smooth positive function in the cylinder $\R \times M_0$. Under some conditions on the transverse manifold $(M_0, g_0)$ such as for instance simplicity\footnote{We say that a compact manifold $(M_0,g_0)$ is simple if any two points in $M_0$ are connected by a unique geodesic depending smoothly on the endpoints and if $\partial M_0$ is strictly convex (its second fundamental form is positive definite).}, the Riemannian manifold $(M,g)$ is said to be \emph{admissible}. In that framework, the authors of \cite{DSFKSU, DSFKLS} were able to determine uniquely the conformal factor $c$ from the knowledge of the DN map at zero frequency $\lambda^2 = 0$, that is, they answered both \textbf{(Q2)} and \textbf{(Q3)} for the class of admissible Riemannian manifolds. Moreover, these results were extended recently to anisotropic Calderon problems with partial data in \cite{KS1} (see below). We also refer to \cite{GT1, Is, IUY1} for other local data results and to the surveys \cite{GT2, KS2} for more references on the subject.

Concerning the anisotropic Calderon problem with data measured on distinct (not necessarily disjoint) sets $\Gamma_D, \Gamma_N$ of $\partial M$, we refer to \cite{KSU} for some positive results of \textbf{(Q3)} in the case of bounded domains $\Omega$ of $\R^n, \ n \geq 3$ equipped with the Euclidean metric. Roughly speaking, in \cite{KSU}, the sets $\Gamma_D, \Gamma_N$ where the measurements are made must overlap a little bit. Precisely, $\Gamma_D \subset \partial \Omega$ can possibly be very small and $\Gamma_N$ must then be slightly larger than $\partial \Omega \setminus \Gamma_D$. We also refer to \cite{KS1} for the generalization of \cite{KSU} to admissible Riemannian manifolds. To explain the result of \cite{KS1}, we recall first that admissible manifolds admit certain functions $\varphi$ which are called \emph{limiting Carleman weights} (LCW) and which are useful for constructing complex geometrical optic solutions. We refer to \cite{DSFKSU} for the definition and properties of limiting Carleman weights on manifolds and their applications. Thanks to the existence of LCW $\varphi$, we can decompose the boundary of $M$ as
$$
  \partial M = \partial M_+ \cup \partial M_{\textrm{tan}} \cup \partial M_-,
$$
where
$$
  \partial M_\pm = \{ x \in \partial M: \ \pm \partial_\nu \varphi(x) > 0 \}, \quad \partial M_{\textrm{tan}} = \{ x \in \partial M: \ \pm \partial_\nu \varphi(x) = 0 \}.
$$
Roughly speaking, the authors of \cite{KS1} show that \textbf{(Q3)} is true\footnote{In fact, additional geometric assumptions on the transverse manifold $(M_0,g_0)$ are needed to give a full proof of this result. We refer to \cite{KS1} Theorem 2.1 for the precise statement.} if the set of Dirichlet data $\Gamma_D$ contains $\partial M_- \cup \Gamma_a$ and the set of Neumann measurements $\Gamma_N$ contains $\partial M_+ \cup \Gamma_a$ where $\Gamma_a$ is some open subset of $\partial M_{\textrm{tan}}$. Hence in particular, the sets $\Gamma_D$ and $\Gamma_N$ must overlap in order to have uniqueness. The only exception occurs in the case where $\partial M_{\textrm{tan}}$ has zero measure. In that case, it is enough to take $\Gamma_D = \partial M_-$ and $\Gamma_N = \partial M_+$ to have uniqueness in \textbf{(Q3)} (see Theorem 2.3 of \cite{KS1}). Note in this case that $\Gamma_D \cap \Gamma_N = \partial M_- \cap \partial M_+ = \emptyset$.

We conclude our survey of anisotropic Calderon problems with the case of data measured on \emph{disjoint sets} for which only a few results are known in the case of zero frequency $\lambda^2 = 0$. We already mentioned above the recent paper \cite{KS1} which concerns a certain subclass of admissible Riemannian manifolds. The only other result we are aware is due to Imanuvilov, Uhlmann and Yamamoto \cite{IUY2} who, roughly speaking, showed in the $2$ dimensional case, that the potential of a Schr\"odinger equation on a two-dimensional domain homeomorphic to a disc, where the boundary is partitioned into eight clockwise-ordered parts $\Gamma_1, \Gamma_2, \dots, \Gamma_8$ is determined by boundary measurements with sources supported on $S = \Gamma_2 \cup \Gamma_6$ and fields observed on $R = \Gamma_4 \cup \Gamma_8$, hence answering \textbf{(Q3)} in this particular setting.

Let us also mention some related papers by Rakesh \cite{Rak} and by Oksanen, Lassas \cite{LO1, LO2} dealing with the \emph{hyperbolic} anisotropic Calderon problem, that is to say in our language, the case where we assume the knowledge of the partial DN map at all frequencies $\lambda^2$. We refer to \cite{KKL} for a thorough account on hyperbolic anisotropic Calderon problem and to \cite{KKLM} for the link between the hyperbolic DN map and the elliptic DN map at all frequencies. For instance, Oksanen and Lassas showed in \cite{LO2} that $(M,g)$ is uniquely determined (up to the gauge invariance (\ref{Inv-Diff})) from the knowledge of $\Lambda_{g,\Gamma_D,\Gamma_N}(\lambda^2)$ at all frequencies $\lambda^2$ under the Hassell-Tao type assumption
$$
  \exists C_0 > 0, \quad \quad \lambda^2_j \leq C_0 \| \partial_\nu \phi_j \|^2_{L^2(\Gamma_D)}, \quad \forall j \geq 1, \dots
$$
where $\lambda^2_j$ are the eigenvalues of $-\Delta_g$ with Dirichlet boundary conditions and $\phi_j$ are the associated normalized eigenfunctions.

Finally, in \cite{Rak},  Rakesh proved that the coefficients of a wave equation on a one-dimensional interval are determined by boundary measurements with sources supported on one end of the interval and the waves observed on the other end. Here again, the uniqueness result entails to know the hyperbolic DN map or equivalently the DN map at all frequencies. \\


In this paper, we are mainly concerned with the anisotropic Calderon problems \textbf{(Q1)} and \textbf{(Q2)} with Dirichlet and Neumann data measured on \emph{disjoint sets}. We provide some simple counterexamples to uniqueness in \textbf{(Q1)} and \textbf{(Q2)} for Riemannian surfaces (in the case of nonzero frequency) and for $3$ dimensional Riemannian manifolds (without restriction on the frequency). In fact, similar counterexamples to uniqueness can be found for any $n$ dimensional Riemannian manifold, but we only give the details in the $3$ dimensional case to keep things concise.

First, we consider a smooth compact Riemannian surface $(\S,g)$ having the topology of a cylinder $M = [0,1] \times T^1$ and that is equipped with a Riemannian metric given in global isothermal coordinates $(x,y)$ by
\begin{equation}\label{Metric-Intro}
  g = f(x) [dx^2 + dy^2].
\end{equation}
Here $f$ is a smooth positive function on $\S$ of the variable $x$ only and $T^1$ stands for the one dimensional torus. The boundary $\partial \S$ of $\S$ is not connected and consists in two copies of $T^1$, precisely
$$
  \partial \S = \Gamma_0 \cup \Gamma_1, \quad \Gamma_0 = \{0\} \times T^1, \quad \Gamma_1 = \{1\} \times T^1.
$$
Let $\Gamma_D$ and $\Gamma_N$ be nonempty open subsets of $\partial M$. We denote $\Lambda_{g,\Gamma_D,\Gamma_N}(\lambda^2)$ the associated partial DN map corresponding to Dirichlet data given on $\Gamma_D$ and Neumann data measured on $\Gamma_N$. We shall prove

\begin{thm} \label{MainThm-2D}
Let $(\S,g)$ denote a Riemannian surface of the form (\ref{Metric-Intro}), \textit{i.e.}
$$
 g = f(x) [dx^2 + dy^2],
$$
with $f$ a smooth positive function on $\S$. Let $\lambda^2 \ne 0$ be a fixed frequency. Let $\Gamma_D$ and $\Gamma_N$ be nonempty open subsets of $\partial M$ that belong to distinct connected components of $\partial M$. Then there exists an infinite dimensional family of non-isometric metrics $\tilde{g_c} = cg$ of the form (\ref{Metric-Intro}), parametrized by smooth positive functions $c = c(x)$, which satisfy $c(0) = c(1) = 1$ such that
$$
	\Lambda_{g,\Gamma_D,\Gamma_N}(\lambda^2) = \Lambda_{\tilde{g_c},\Gamma_D,\Gamma_N}(\lambda^2).
$$
\end{thm}

We emphasize that the family of metrics $\tilde{g_c}$ parametrized by the smooth positive functions $c = c(x)$ satisfying $c(0) = c(1) = 1$ is \emph{explicit}. We refer to the end of the proof of Theorem \ref{MainThm-T-2D} below for its complete description and more particularly to the formulae (\ref{Iso3}) and (\ref{Iso4}) for a countable family of examples. Conversely, at zero frequency $\lambda^2 = 0$ and still for Dirichlet and Neumann data measured on distinct connected components of $\partial M$, we show that a metric $g = f(x) [dx^2 + dy^2]$ can be uniquely determined by $\Lambda_{g,\Gamma_D,\Gamma_N}(0)$ up to the gauge invariance (\ref{Inv-Conf}). We mention finally that in the case where the data $\Gamma_D$ and $\Gamma_N$ belong to the same connected component of $\partial M$, we can show that a metric $g = f(x) [dx^2 + dy^2]$ can be uniquely determined by $\Lambda_{g,\Gamma_D,\Gamma_N}(\lambda^2)$ for any frequency $\lambda^2$ (up to the gauge invariance (\ref{Inv-Conf}) when $\lambda^2 = 0$). This is proved in Theorem \ref{MainThm-R-2D}.

In which sense is our family of metrics $\tilde{g_c}$ a counterexample to \textbf{(Q1)} and \textbf{(Q2)} when the data are measured on disjoint sets? Since we work at a fixed nonzero frequency $\lambda^2 \ne 0$, the only gauge invariance for the above partial anisotropic Calderon problem is a priori (\ref{Inv-Diff}), \textit{i.e.} $g$ and $\tilde{g}_c$ coincide up to the existence of a diffeomorphism $\phi: \, M \longrightarrow M$ such that $\phi_{|\Gamma_D \cup \Gamma_N} = Id$ and $\phi^* g = \tilde{g}_c = cg$.

Assume the existence of such a diffeomorphism $\phi$ in the case $\Gamma_D \subset \Gamma_0$ and $\Gamma_N \subset \Gamma_1$. Then $\phi$ must send the boundary $\partial M$ into itself, \textit{i.e.} $\phi(\partial M) = \partial M$. More precisely, since by hypothesis $\phi_{|\Gamma_D} = Id$ and $\phi_{|\Gamma_N} = Id$, we see that in fact $\phi$ must send the connected components $\Gamma_0$ and $\Gamma_1$ of the boundary into themselves, \textit{i.e.} $\phi(\Gamma_0) = \Gamma_0$ and $\phi(\Gamma_1) = \Gamma_1$. Let us denote now by $g_0 = f(0) dy^2$ and $g_1 = f(1) dy^2$ the metrics on $\Gamma_0$ and $\Gamma_1$ induced by $g$. Taking the restrictions to $\Gamma_0$ and $\Gamma_1$ of $\phi^* g = \tilde{g}_c = cg$, using $c(0) = c(1) = 1$ and our previous result, we get $\phi^* g_0 = g_0$ and $\phi^* g_1 = g_1$. In other words, $\phi_{|\Gamma_0}$ and $\phi_{|\Gamma_1}$ are isometries of $(\Gamma_0,g_0)$ and $(\Gamma_1,g_1)$ respectively. But the isometries of $(\Gamma_0,g_0)$ and $(\Gamma_1,g_1)$ are simply the transformations $y \mapsto \pm y + a$ for any constant $a$. Using our hypothesis $\phi_{|\Gamma_D} = Id$ and $\phi_{|\Gamma_N} = Id$ again, we see that the unique possibility is $\phi_{|\Gamma_0} = \phi_{|\Gamma_1} = Id$. We thus have $\phi_{|\partial M} = Id$. According to \cite{Li}, the only diffeomorphism $\phi$ satisfying the previous properties is $Id$. We thus conclude that our family of metrics $\tilde{g}_c = c g$ cannot come from the pull back of the initial metric $g$ by such a diffeomorphism and thus provide the claimed counterexamples. \\

We also solve the anisotropic Calderon problem \textbf{(Q3)} in the class of smooth compact Riemannian surfaces (\ref{Metric-Intro}) for potentials $V \in L^\infty$ that only depend on the variable $x$. We show similarly that if the Dirichlet and Neumann data $\Gamma_D$ and $\Gamma_N$ belong to two distinct components of $\partial \S$, then there exists an infinite dimensional family of potentials $\tilde{V} = \tilde{V}(x)$ that satisfies
$$
	\Lambda_{g,V,\Gamma_D,\Gamma_N}(\lambda^2) = \Lambda_{g,\tilde{V},\Gamma_D,\Gamma_N}(\lambda^2).
$$
This is done in Theorem \ref{MainThm-Schrodinger-2D}.

In $3$ dimensions, we consider the family of smooth compact Riemannian manifolds $(M,g)$ which have the topology of a toric cylinder $M = [0,1] \times T^2$ and that are equipped with a metric
\begin{equation} \label{Metric-3D-Intro}
  g = f(x) dx^2 + f(x) dy^2 + h(x) dz^2,
\end{equation}
where $f, h$ are smooth positive functions on $M$. Once again, the boundary $\partial M$ of $M$ is disconnected and consists in two copies of $T^2$, precisely
$$
  \partial M = \Gamma_0 \cup \Gamma_1, \quad \Gamma_0 = \{0\} \times T^2, \quad \Gamma_1 = \{1\} \times T^2.
$$
Let $\Gamma_D$ and $\Gamma_N$ be nonempty open subsets of $\partial M$. We denote $\Lambda_{g,\Gamma_D,\Gamma_N}(\lambda^2)$ the associated partial DN map corresponding to Dirichlet data given on $\Gamma_D$ and Neumann data measured on $\Gamma_N$. We shall prove

\begin{thm} \label{MainThm-3D}
Let $(M,g)$ and $(M,\tilde{g})$ denote two generic Riemannian manifolds of the form (\ref{Metric-3D-Intro}), \textit{i.e.}
$$
	g = f(x) dx^2 + f(x) dy^2 + h(x) dz^2, \quad \tilde{g} = \tilde{f}(x) dx^2 + \tilde{f}(x) dy^2 + \tilde{h}(x) dz^2.
$$
Let $\lambda^2 \in \R$ be a fixed frequency and let $\Gamma_D$ and $\Gamma_N$ be nonempty open subsets of $\partial M$ such that $\Gamma_D \cap \Gamma_N = \emptyset$. Then \\

\noindent 1) There exist infinitely many pairs of non isometric Riemannian manifolds $(M,g)$ and $(M,\tilde{g})$ with $\tilde{g} = c g$ for some smooth positive strictly increasing or decreasing functions $c = c(x)$ satisfying
\begin{itemize}
\item $c(0) = 1$ if $\Gamma_D, \Gamma_N \subset \Gamma_0$,
\item $c(1) = 1$ if $\Gamma_D, \Gamma_N \subset \Gamma_1$,
\item $c(1)^3 = c(0) \ne 1$ if $\Gamma_D \subset \Gamma_0$ and $\Gamma_N \subset \Gamma_1$,
\item $c(0)^3 = c(1) \ne 1$ if $\Gamma_D \subset \Gamma_1$ and $\Gamma_N \subset \Gamma_0$,
\end{itemize}
such that
$$
	\Lambda_{g,\Gamma_D,\Gamma_N}(\lambda^2) = \Lambda_{\tilde{g},\Gamma_D,\Gamma_N}(\lambda^2).
$$

\noindent 2) If moreover $\lambda^2 = 0$, there exists a one parameter family of Riemannian manifolds $(M,\tilde{g}_a)_{a > 0}$ non isometric with the given Riemannian manifold $(M,g)$, satisfying
$$
	\Lambda_{g,\Gamma_D,\Gamma_N}(0) = \Lambda_{\tilde{g}_a,\Gamma_D,\Gamma_N}(0).
$$
The one parameter family of metrics $(\tilde{g}_a)_{a > 0}$ has the form $\tilde{g}_a = c_a g$ where $c_a = c(x,a)$ are smooth positive strictly increasing or decreasing functions satisfying
\begin{itemize}
\item $c(0,a) = 1$ if $\Gamma_D, \Gamma_N \subset \Gamma_0$,
\item $c(1,a) = 1$ if $\Gamma_D, \Gamma_N \subset \Gamma_1$,
\item $c(x,a) > 1$ or $c(x,a) < 1$ if $\Gamma_D \subset \Gamma_0$ (resp. $\Gamma_D \subset \Gamma_1)$ and $\Gamma_N \subset \Gamma_1$ (resp. $\Gamma_D \subset \Gamma_0)$,
\end{itemize}
that are explicitly given in terms of $g$.
\end{thm}

As in the $2$ dimensional case, the only gauge invariance for the above partial anisotropic Calderon problems \textbf{(Q1)} and \textbf{(Q2)} with data measured on disjoint sets is a priori (\ref{Inv-Diff}). Assume that there exists a diffeomorphism $\phi: \, M \longrightarrow M$ such that $\phi_{|\Gamma_D \cup \Gamma_N} = Id$ and $\phi^* g = \tilde{g} = cg$. In particular, since $\phi$ is a diffeomorphism, we see that $Vol_g(M) = Vol_{\phi^*g}(M) = Vol_{cg}(M)$. Hence we must have
$$
  \int_M \sqrt{|g|} dx dy dz = \int_M c^{3/2} \sqrt{|g|} dx dy dz.
$$
But in any case, we know that $c > 1$ or $c < 1$ on $(0,1)$. Hence the above equality is impossible. We conclude that our family of metrics is not captured by this gauge invariance and thus provides counterexamples to uniqueness.

We emphasize that we could extend the results of Theorem \ref{MainThm-3D} to higher dimensional Riemannian manifolds. Consider for instance $n$ dimensional Riemannian manifolds $(M,g)$ which have the topology of a toric cylinder $M = [0,1] \times T^{n-1}$ and that are equipped with a metric
\begin{equation} \label{Metric-nD-Intro}
  g = f(x) dx^2 + f(x) dy_1^2 + h_2(x) dy_2^2 + \dots + h_{n-1}(x) dy_{n-1}^2,
\end{equation}
where $f, h_2, \dots, h_{n-1}$ are smooth positive functions on $M$. The analysis given for the 3 dimensional models extend in a straightforward way to such Riemannian manifolds and the results are basically the same.

We also solve the anisotropic Calderon problem \textbf{(Q3)} in the class of smooth compact Riemannian manifolds (\ref{Metric-3D-Intro}) for potentials $V \in L^\infty$ that only depend on the variable $x$. Contrary to Theorem \ref{MainThm-3D}, we show
in Theorems \ref{MainThm-Schrodinger-2D} and \ref{MainThm-Schrodinger-3D} that if the Dirichlet and Neumann data $\Gamma_D$ and $\Gamma_N$ belong to the same connected component of $\partial M$ and if $	\Lambda_{g,V,\Gamma_D,\Gamma_N}(\lambda^2) = \Lambda_{g,\tilde{V},\Gamma_D,\Gamma_N}(\lambda^2)$, (and a technical asumption on $\Gamma_N$  in Theorem \ref{MainThm-Schrodinger-3D}) , then
$$
	V = \tilde{V}.
$$
The anisotropic Calderon problems \textbf{(Q2)} and \textbf{(Q3)} are thus not equivalent in our 3 dimensional models. Moreover, if the Dirichlet and Neumann data $\Gamma_D$ and $\Gamma_N$ does not belong to the same connected component of $\partial M$, we show there exists an explicit infinite dimensional family of potentials $\tilde{V}$ that satisfies $\Lambda_{g,V,\Gamma_D,\Gamma_N}(\lambda^2) = \Lambda_{g,\tilde{V},\Gamma_D,\Gamma_N}(\lambda^2)$.

The remainder of this paper is devoted to the proofs of Theorems \ref{MainThm-2D} and \ref{MainThm-3D}. In Section \ref{Preliminary}, we collect some results on the inverse spectral problem of one dimensional Schr\"odinger operators. In particular, we define the notion of characteristic and Weyl-Titchmarsh functions associated with one dimensional equation Schr\"odinger equations with certain boundary conditions and state a version of the Borg-Marchenko Theorem that will be needed for our proof. We refer for instance to \cite{Be, ET, GS, KST, Lev, Te} for a detailed account of these results.

The interest in one dimensional inverse spectral theory comes from the presence of symmetries in our models. These symmetries allow to decompose the solution of the Dirichlet problem and the DN map onto the Fourier modes $\{e^{imy}\}_{m \in \Z}$ in 2D and the Fourier modes $\{e^{imy} e^{inz}\}_{m,n \in \Z}$ in 3D. Hence we are able to reduce our initial problem to an infinite number of 1D inverse spectral problems for which the results recalled in Section \ref{Preliminary} will be useful.

In Section \ref{2D}, we describe the $2$ dimensional models and construct first the global and partial DN maps that we aim to study as well as their restrictions onto the Fourier modes $\{e^{imy}\}_{m \in \Z}$. Using essentially the Complex Angular Momentum Method (see \cite{Re, Ra, DKN1, DN3, DN4, DGN}) that consists in allowing the angular momentum $m$ to be complex, we solve in Theorem \ref{MainThm-R-2D} an anisotropic Calderon problem with Dirichlet and Neumann data measured on some open subsets $\Gamma_D$ and $\Gamma_N$ of $\partial M$. We continue with the proof of our main Theorem \ref{MainThm-2D} in two dimensions. We finish this Section solving the anisotropic Calderon problem \textbf{(Q3)} in Theorem \ref{MainThm-Schrodinger-2D} and giving new counterexamples to uniqueness in this setting.

In Section \ref{3D}, we perform a similar analysis for the $3$ dimensional models. We construct first the global and partial DN maps that we aim to study as well as their restrictions onto the Fourier modes $\{e^{imy + inz}\}_{m,n \in \Z}$. In Theorem \ref{MainThm-R-3D}, we solve the anisotropic Calderon problem with Dirichlet and Neumann data measured on some overlaping open subsets $\Gamma_D$ or $\Gamma_N$ belonging to the same connected component of $\partial M$. Once again here, the main tool is the Complex Angular Momentum method that makes possible to complexify the angular momenta $m$ and $n$. We then proceed to prove our second main Theorem \ref{MainThm-3D} in three dimensions. We finish this Section showing that the anisotropic Calderon problem \textbf{(Q3)} has a unique solution in the case where the Dirichlet and Neumann data belong to the same connected component of the boundary. This is done in Theorem \ref{MainThm-Schrodinger-3D}.


\Section{Preliminary results on 1D inverse spectral problems} \label{Preliminary}

We consider the class of ODE on the interval $[0,1]$ given by
\begin{equation} \label{Equation}
  -v'' + q(x) v = - \mu^2 v,  \quad q \in L^2([0,1]), \quad q \ \textrm{real},
\end{equation}
with Dirichlet boundary conditions
\begin{equation} \label{BC}
  v(0) = 0, \quad v(1) = 0.
\end{equation}

Since the potential $q$ belongs to $L^2([0,1]) \subset L^1([0,1])$, we can define for all $\mu \in \C$ two fundamental systems of solutions (FSS)
$$
  \{ c_0(x,\mu^2), s_0(x,\mu^2)\}, \quad \{ c_1(x,\mu^2), s_1(x,\mu^2)\},
$$
of (\ref{Equation}) by imposing the Cauchy conditions
\begin{equation} \label{FSS}
  \left\{ \begin{array}{cccc} c_0(0,\mu^2) = 1, & c_0'(0,\mu^2) = 0, & s_0(0,\mu^2) = 0, & s_0'(0,\mu^2) = 1, \\
 	                  c_1(1,\mu^2) = 1, & c'_1(1,\mu^2) = 0, & s_1(1,\mu^2) = 0, & s'_1(1,\mu^2) = 1. \end{array} \right.
\end{equation}
It follows from (\ref{FSS}) that
\begin{equation} \label{Wronskian-FSS}
  W(c_0, s_0) = 1, \quad W(c_1, s_1) = 1, \quad \forall \mu \in \C,
\end{equation}
where $W(u,v) = uv' - u'v$ is the Wronskian of $u,v$. Moreover, the FSS $\{ c_0(x,\mu^2), s_0(x,\mu^2)\}$ and $\{ c_1(x,\mu^2), s_1(x,\mu^2)\}$ are entire functions with respect to the variable $\mu^2 \in \C$ by standard properties of ODEs depending analytically on a parameter.

We then define the characteristic function of (\ref{Equation}) with boundary conditions (\ref{BC}) by
\begin{equation} \label{Char}
  \Delta(\mu^2) = W(s_0, s_1).
\end{equation}
The characteristic function $\mu^2 \mapsto \Delta(\mu^2)$ is entire with respect to $\mu^2$ and its zeros $(\alpha^2_{k})_{k \geq 1}$ correspond to "minus" the eigenvalues of the selfadjoint operator $-\frac{d^2}{dx^2} + q$ with Dirichlet boundary conditions. The eigenvalues $\alpha_k^2$ are thus real and satisfy $-\infty < \dots < \alpha^2_2 < \alpha_1^2 < \infty$.

We next define two Weyl-Titchmarsh functions by the following classical prescriptions. Let the Weyl solutions $\Psi$ and $\Phi$ be the unique solutions of (\ref{Equation}) having the form
\begin{equation} \label{WeylFunction}
  \Psi(x,\mu^2) = c_0(x,\mu^2) + M(\mu^2) s_0(x,\mu^2), \quad \Phi(x,\mu^2) = c_1(x,\mu^2) - N(\mu^2) s_1(x,\mu^2),
\end{equation}
which satisfy the Dirichlet boundary condition at $x = 1$ and $x=0$ respectively. Then a short calculation using (\ref{FSS}) shows that the Weyl-Titchmarsh functions $M(\mu^2)$ and $N(\mu^2)$ are uniquely defined by
\begin{equation} \label{WT}
  M(\mu^2) = - \frac{W(c_0, s_1)}{\Delta(\mu^2)}, \quad N(\mu^2) = -\frac{W(c_1, s_0)}{\Delta(\mu^2)}.
\end{equation}
For later use, we introduce the functions $D(\mu^2) = W(c_0, s_1)$ and $E(\mu^2) = W(c_1, s_0)$ which also turn out to be entire functions of $\mu^2$. We then have the following notation for the Weyl-Titchmarsh functions
$$
  M(\mu^2) = - \frac{D(\mu^2)}{\Delta(\mu^2)}, \quad N(\mu^2) = - \frac{E(\mu^2)}{\Delta(\mu^2)}.
$$

We now collect some analytic results involving the functions $\Delta(\mu^2), D(\mu^2), E(\mu^2), M(\mu^2), N(\mu^2)$ and give a version of the Borg-Marchenko Theorem we shall need later.

\begin{lemma} \label{Asymp}
  The FSS $\{ c_0(x,\mu^2), s_0(x,\mu^2)\}$ and $\{ c_1(x,\mu^2), s_1(x,\mu^2)\}$ have the following asymptotics uniformly with respect to $x \in [0,1]$ as $|\mu| \to \infty$ in the complex plane $\C$.
	\begin{equation} \label{Asymp-0}
	  \left\{ \begin{array}{ccc}
		  c_0(x,\mu^2) & = & \cosh(\mu x) + O \left( \frac{e^{|\Re(\mu)| x}}{\mu} \right), \\
			c'_0(x,\mu^2) & = & \mu \sinh(\mu x) + O \left( e^{|\Re(\mu)| x} \right), \\
			s_0(x,\mu^2) & = & \frac{\sinh(\mu x)}{\mu} + O \left( \frac{e^{|\Re(\mu)| x}}{\mu^2} \right), \\
			s'_0(x,\mu^2) & = & \cosh(\mu x) + O \left( \frac{e^{|\Re(\mu)| x}}{\mu} \right),
		\end{array} \right.
	\end{equation}
	and
	\begin{equation} \label{Asymp-1}
	  \left\{ \begin{array}{ccc}
		  c_1(x,\mu^2) & = & \cosh(\mu (1-x)) + O \left( \frac{e^{|\Re(\mu)| (1-x)}}{\mu} \right), \\
			c'_1(x,\mu^2) & = & -\mu \sinh(\mu (1-x)) + O \left( e^{|\Re(\mu)| (1-x)} \right), \\
			s_1(x,\mu^2) & = & -\frac{\sinh(\mu (1-x))}{\mu} + O \left( \frac{e^{|\Re(\mu)| (1-x)}}{\mu^2} \right), \\
			s'_1(x,\mu^2) & = & \cosh(\mu (1-x)) + O \left( \frac{e^{|\Re(\mu)| (1-x)}}{\mu} \right).
		\end{array} \right.
	\end{equation}
\end{lemma}
\begin{proof}
  These asymptotics are classical and can be found in \cite{PT}, Theorem 3, p. 13.
\end{proof}

\begin{coro} \label{OrderHalf}
  1. For each fixed $x \in [0,1]$, the FSS $\{ c_0(x,\mu^2), s_0(x,\mu^2)\}$ and $\{ c_1(x,\mu^2), s_1(x,\mu^2)\}$ are entire functions of order $\frac{1}{2}$ with respect to the variable $\mu^2$. \\
	2. The characteristic function $\Delta(\mu^2)$ and the functions $D(\mu^2)$ and $E(\mu^2)$ are entire functions of order $\frac{1}{2}$ with respect to the variable $\mu^2$. \\
	3. We have the following asymptotics in the complex plane $\C$:
	$$
	  \left\{ \begin{array}{c} \Delta(\mu^2) = \frac{\sinh{\mu}}{\mu} + O \left( \frac{e^{|\Re(\mu)|}}{\mu^2} \right), \\
		D(\mu^2) = \cosh (\mu) +  O \left( \frac{e^{|\Re(\mu)|}}{\mu} \right), \quad E(\mu^2) = \cosh (\mu) +  O \left( \frac{e^{|\Re(\mu)|}}{\mu} \right).
		\end{array} \right.
	$$
In particular, $M(\mu^2) = \mp \mu + O(1), \quad N(\mu^2) = \mp \mu + O(1)$ when $\mu \to \pm \infty,$ $\mu \in \R$.
\end{coro}
\begin{proof}
  The proof of 1., 2. and 3. follows directly from (\ref{Char}), (\ref{WT}) and Lemma \ref{Asymp}.
\end{proof}

\begin{coro} \label{Hadamard}
  The characteristic function $\Delta(\mu^2)$ and the functions $D(\mu^2)$ and $E(\mu^2)$ can be written as
	\begin{equation} \label{HadFacto}
	  \begin{array}{ccc}
		  \Delta(\mu^2) & = & A \mu^{2p} \prod_{k = 1}^\infty \left( 1 - \frac{\mu^2}{\alpha_{k}^2} \right), \\
			D(\mu^2) & = & B \mu^{2q} \prod_{k = 1}^\infty \left( 1 - \frac{\mu^2}{\beta_{k}^2} \right), \\
			E(\mu^2) & = & C \mu^{2r} \prod_{k = 1}^\infty \left( 1 - \frac{\mu^2}{\gamma_{k}^2} \right),
		\end{array}
	\end{equation}
	where $(\alpha_{k}^2)_{k \geq 1}$, $(\beta_{k}^2)_{k \geq 1}$ and $(\gamma_{k}^2)_{k \geq 1}$ are the zeros of the entire functions $\Delta(\mu^2)$, $D(\mu^2)$ and $E(\mu^2)$ respectively, $A,B,C$ are constants and $p,q,r$ are the multiplicities of the root at the origin.
\end{coro}
\begin{proof}
  This is a direct consequence of Hadamard's factorization Theorem (see \cite{Boa, Lev, Ru}) for the entire functions $\Delta(\mu^2)$, $D(\mu^2)$ and $E(\mu^2)$ of order $\frac{1}{2}$.
\end{proof}

We shall also need the Cartwright class of analytic functions.

\begin{defi} \label{Cartwright}
  We say that an entire function $f$ belongs to the Cartwright class $\mathcal{C}$ if $f$ is of exponential type (\textit{i.e.} $\forall z \in \C, \quad |f(z)| \leq C e^{A|z|}$ for some positive constants $A, C$) and satisfies
	$$
	  \int_{\R} \frac{\log^+(|f(iy)|)}{1+y^2} dy < \infty,
	$$
	where
	$$
	  \log^+(x) = \left\{ \begin{array}{cc} \log(x) & \log(x) \geq 0, \\ 0, & \log(x) < 0. \end{array} \right.
	$$
\end{defi}

\begin{rem} \label{Nevanlinna}
Note that if we restrict the Cartwright class $\mathcal{C}$ to functions analytic on the half plane $\C^+ = \{ z \in \C, \ \Re(z) \geq 0 \}$, then it reduces to the Nevanlinna class $N^+(\C^+)$ used in \cite{DN3, DGN} (see \cite{Lev}, Remark, p 116).
\end{rem}

It is well known that the zeros of entire functions in the Cartwright class $\mathcal{C}$ have a certain distribution in the complex plane (see \cite{Lev}, Theorem 1, p 127). As a consequence, the following uniqueness Theorem holds.

\begin{prop} \label{UniquenessC}
  Let $f \in \mathcal{C}$ and $(z_n)_n$, $z_n \not=0$, be the zeros  of $f$. If there exists a subset $\mathcal{L} \subset \N$ such that
	$$
	  \sum_{n \in \mathcal{L}} \frac{\Re(z_n)}{|z_n|^2} = \infty,
	$$
	then $f = 0$ on $\C$. In particular, if the function $f\in \mathcal{C}$ vanishes on the set of integers $\N$, then $f$ is identically $0$ on $\mathcal{C}$.
\end{prop}

Thanks to Lemma \ref{Asymp}, we can apply these results to the entire functions
\begin{equation} \label{Smalld}
  \mu \mapsto \delta(\mu) := \Delta(\mu^2), \quad \mu \mapsto d(\mu) := D(\mu^2), \quad \mu \mapsto e(\mu) := E(\mu^2).
\end{equation}
Precisely, we can prove

\begin{coro} \label{Cart}
  The functions $\mu \mapsto \delta(\mu)$, $\mu \mapsto d(\mu)$ and $\mu \mapsto e(\mu)$ belong to the Cartwright class $\mathcal{C}$.
\end{coro}
\begin{proof}
  Thanks to Lemma \ref{Asymp}, the functions $\mu \mapsto \delta(\mu)$, $\mu \mapsto d(\mu)$ and $\mu \mapsto e(\mu)$ are clearly of exponential type and are bounded on the imaginary axis $i\R$. Hence they belong to the Cartwright class $\mathcal{C}$.
\end{proof}

Let us finally recall here the celebrated Borg-Marchenko Theorem which roughly speaking states that the Weyl-Titchmarsh functions $M$ or $N$ of the differential expression $-\frac{d^2}{d x^2} + q$ with regular boundary condition at $x=0$ and $x=1$ and with a real-valued potential $q$ determines uniquely this potential. We refer for instance to \cite{Be, Bo1, Bo2, ET, FY, GS, KST, Ma, Te} for historic and recent developments on the theory of WT functions. We state here a version from \cite{ET} (see Corollary 4.3.) adapted to our particular problem.
\begin{thm}[Borg-Marchenko] \label{BM}
  Consider two boundary value problems (\ref{Equation}) with potentials $q$ and $\tilde{q}$ in $L^1([0,1])$. Let $M(\mu^2)$ and $\tilde{M}(\mu^2)$ the Weyl-Titchmarsh functions associated to (\ref{Equation}) using (\ref{WT}). Then, if
$$
  M(\mu^2) = \tilde{M}(\mu^2) + f(\mu^2), \quad \mu^2 \in \C \setminus \R,
$$	
for any $f$ entire function of growth order at most $\frac{1}{2}$ (here the equality is understood on the domains of holomorphy of both sides of the equality, that is in this case, on $\C$ minus the set of discrete eigenvalues of (\ref{Equation}) that lie on the real axis), then
$$
  q = \tilde{q}, \quad \textrm{on} \ [0,1].
$$
Conversely, if $q = \tilde{q}$ on $[0,1]$, then obviously $M = \tilde{M}$ on $\C \setminus \R$. \\
Of course, the same results hold with $M(\mu^2)$ and $\tilde{M}(\mu^2)$ replaced by $N(\mu^2)$ and $\tilde{N}(\mu^2)$.	
\end{thm}


\Section{The two dimensional case} \label{2D}

\subsection{The model and the Dirichlet-to-Neumann map} \label{Model-DNmap-2D}

In this Section, we work on a Riemannian surface $(\S,g)$ which has the topology of a cylinder $M = [0,1] \times T^1$ and that is equipped with a Riemannian metric given in global isothermal coordinates $(x,y)$ by
\begin{equation} \label{Metric-2D}
  g = f(x) [dx^2 + dy^2].
\end{equation}
Here $f$ is a smooth positive function on $\S$ of the variable $x$ only and $T^1$ stands for the circle. The metric $g$ obviously possesses one Killing vector field $\partial_y$ that generates the cylindrical symmetry of our surface of revolution. The boundary $\partial \S$ of $\S$ consists in two copies of $T^1$, precisely
$$
  \partial \S = \Gamma_0 \cup \Gamma_1, \quad \Gamma_0 = \{0\} \times T^1, \quad \Gamma_1 = \{1\} \times T^1.
$$

In our global coordinates system $(x,y)$, the positive Laplace-Beltrami operator takes the expression
$$
  -\Delta_g = \frac{1}{f} \left( -\partial_x^2 - \partial_y^2 \right).
$$
It is well known that the Laplace-Beltrami operator $-\Delta_g$ with Dirichlet boundary conditions is selfadjoint on $L^2(\S, dVol)$ and has pure point spectrum $\{ \lambda_j^2\}_{j \geq 1}$ so that $0 < \lambda_1^2 < \lambda^2_2 \leq \dots \leq \lambda_j^2 \to +\infty$.

Consider the Dirichlet problem at a frequency $\lambda^2$ such that $\lambda \notin \{ \lambda_j^2\}_{j \geq 1}$. That is we look at the solutions $u$ of the following PDE
$$
  -\Delta_g u = \lambda^2 u, \quad \textrm{on} \ \S,
$$
which can be rewritten in our coordinates system as
\begin{equation} \label{Eq1-2D}
  - \partial_x^2 u - \partial_y^2 u - \lambda^2 f u = 0, \quad \textrm{on} \ \S,
\end{equation}
with the boundary conditions
\begin{equation} \label{BC1-2D}
  u = \psi \quad \textrm{on} \ \partial \S.
\end{equation}
The DN map $\Lambda_g(\lambda^2)$ is then defined by (\ref{DN-Abstract}) as
\begin{equation} \label{DN-2D}
  \Lambda_g(\lambda^2) (\psi) = \left( \partial_\nu u \right)_{|\partial \S},
\end{equation}
where $u$ is the unique solution in $H^1(\S)$ of (\ref{Eq1-2D}) and $\left( \partial_\nu u \right)_{|\partial \S} \in H^{-1/2}(\partial \S)$ is the weak normal derivative of $u$ on $\partial \S$.

In order to find a nice representation of the DN map, we introduce certain notations. Since the boundary $\partial \S$ of $\S$ has two disjoint components $\partial \S = \Gamma_0 \cup \Gamma_1$, we may decompose the Sobolev spaces $H^s(\partial \S)$ as $H^s(\partial \S) = H^s(\Gamma_0) \oplus H^s(\Gamma_1)$ for any $s \in \R$. We also recall that $\Gamma_0 = \Gamma_1 = T^1$. Moreover, we shall use the vector notation
$$
 \varphi = \left( \begin{array}{c} \varphi^0 \\ \varphi^1 \end{array} \right),
$$
for all elements $\varphi$ of $H^s(\partial \S) = H^s(\Gamma_0) \oplus H^s(\Gamma_1)$. Finally, since the DN map is a linear operator from $H^{1/2}(\partial \S)$ to $H^{-1/2}(\partial \S)$, it has the structure of an operator valued $2 \times 2$ matrix
$$
  \Lambda_g(\lambda^2) = \left( \begin{array}{cc} L(\lambda^2) & T_R(\lambda^2) \\ T_L(\lambda^2) & R(\lambda^2) \end{array} \right),
$$
where $L(\lambda^2), R(\lambda^2), T_R(\lambda^2), T_L(\lambda^2)$ are understood as operators from $H^{1/2}(T^1)$ to $H^{-1/2}(T^1)$. The operators $L(\lambda^2), R(\lambda^2)$ correspond to the partial DN map whose measurements are restricted to $\Gamma_0$ and $\Gamma_1$ respectively whereas the operators $T_R(\lambda^2), T_L(\lambda^2)$ correspond to the partial DN maps whose measurements are made on the disjoints sets $\Gamma_0$ and $\Gamma_1$. For instance, $T_L(\lambda^2)$ is the operator from $H^{1/2}(\Gamma_0)$ to $H^{-1/2}(\Gamma_1)$ given by
$$
  T_L(\lambda^2)(\psi^0) = \left( \partial_\nu u \right)_{|\Gamma_1},
$$
where $u$ is the unique solution of (\ref{Eq1-2D}) - (\ref{BC1-2D}) such that $supp \,\psi^0 \subset \Gamma_0$.

To summarize, using the notations in the Introduction, we have the following dictionary
$$
  L(\lambda^2) = \Lambda_{g,\Gamma_0}(\lambda^2), \quad R(\lambda^2) = \Lambda_{g,\Gamma_1}(\lambda^2),
$$
$$	
	T_L(\lambda^2) = \Lambda_{g,\Gamma_0, \Gamma_1}(\lambda^2), \quad T_R(\lambda^2) = \Lambda_{g,\Gamma_1, \Gamma_0}(\lambda^2).
$$

Now we take advantage of the cylindrical symmetry of $(\S,g)$ to find a simple expression of the DN map. We first write $\psi = (\psi^0, \psi^1) \in H^{1/2}(\Gamma_0) \times H^{1/2}(\Gamma_1)$ using their Fourier series representation as
$$
 \psi^0 = \sum_{m \in \Z} \psi^0_{m} Y_{m}, \quad \psi^1 = \sum_{m \in \Z} \psi^1_{m} Y_{m},
$$
where
$$
  Y_{m}(y) = e^{imy}.
$$
Note that for any $s \in \R$, the space $H^{s}(T^1)$ can be described as
$$
  \varphi \in H^{s}(T^1) \ \Longleftrightarrow \ \left\{ \varphi \in \D'(T^1), \ \varphi = \sum_{m \in \Z} \varphi_{m} Y_{m}, \quad \sum_{m \in \Z} (1 + m^2)^{s} |\varphi_{m}|^2 < \infty \ \right\}.
$$
Therefore, the unique solution $u$ of (\ref{Eq1-2D}) - (\ref{BC1-2D}) can be looked for under the form
$$
  u = \sum_{m \in \Z} u_{m}(x) Y_{m}(y),
$$
where for all $m \in \Z$, the functions $u_{m}$ are the unique solutions of the ODEs (w.r.t. $x$) with boundary conditions
\begin{equation} \label{Eq2-2D}
  \left\{ \begin{array}{c} -u_{m}'' + m^2 u_{m} - \lambda^2 f(x) u_{m} = 0, \\
    u_{m}(0) = \psi^0_{m}, \quad u_{m}(1) = \psi^1_{m}. \end{array} \right.
\end{equation}

The DN map can now be diagonalized on the Hilbert basis $\{ Y_{m} \}_{m \in \Z}$ and can be shown to have a very simple expression on each Fourier mode. First, a short calculation using the particular form of the metric (\ref{Metric-2D}) and for smooth enough Dirichlet data $\psi$ shows that
$$
  \Lambda_g(\lambda^2) \left( \begin{array}{c} \psi^0 \\ \psi^1 \end{array} \right) = \left( \begin{array}{c} \left( \partial_\nu u \right)_{|\Gamma_0} \\ \left( \partial_\nu u \right)_{|\Gamma_1} \end{array} \right) = \left( \begin{array}{c} - \frac{1}{\sqrt{f(0)}} (\partial_x u)_{|x=0} \\ \frac{1}{\sqrt{f(1)}} (\partial_x u)_{|x=1} \end{array} \right).
$$
Hence, if we let the DN map act on the vector space generated by the Fourier mode $Y_{m}$, we get
\begin{equation} \label{DN-Partiel-0-2D}
  \Lambda_g(\lambda^2) \left( \begin{array}{c} \psi_{m}^0 Y_{m} \\ \psi_{m}^1 Y_{m} \end{array} \right) = \left( \begin{array}{c} - \frac{1}{\sqrt{f(0)}} u_{m}'(0) Y_{m} \\ \frac{1}{\sqrt{f(1)}} u_{m}'(1) Y_{m} \end{array} \right).
\end{equation}
We denote by
$$
  \Lambda_g(\lambda^2)_{|<Y_{m}>} = \Lambda^{m}_g(\lambda^2) = \left( \begin{array}{cc} L^{m}(\lambda^2) & T^{m}_R(\lambda^2) \\ T^{m}_L(\lambda^2) & R^{m}(\lambda^2) \end{array} \right),
$$
the $2 \times 2$ matrix corresponding to the restriction of the DN map to each Fourier mode $<Y_{m}>$. This operator is clearly defined  for all $m \in \Z$ and for all $\left(\psi_{m}^0, \psi_{m}^1 \right) \in \C \oplus \C$ by
\begin{equation} \label{DN-Partiel-1-2D}
  \Lambda^{m}_g(\lambda^2) \left( \begin{array}{c} \psi_{m}^0  \\ \psi_{m}^1  \end{array} \right) =  \left( \begin{array}{cc} L^{m}(\lambda^2) & T^{m}_R(\lambda^2) \\ T^{m}_L(\lambda^2) & R^{m}(\lambda^2) \end{array} \right) \left( \begin{array}{c} \psi_{m}^0  \\ \psi_{m}^1  \end{array} \right) = \left( \begin{array}{c} - \frac{1}{\sqrt{f(0)}} u_{m}'(0) \\ \frac{1}{\sqrt{f(1)}} u_{m}'(1) \end{array} \right).
\end{equation}

We can further simplify the partial DN maps $\Lambda^{m}_g(\lambda^2)$ by interpreting their coefficients as the characteristic and Weyl-Titchmarsh functions of the 1D equation (\ref{Eq2-2D}) with appropriate boundary conditions. Precisely, consider the ODE
\begin{equation} \label{Eq4-2D}
  -v'' + q_{\lambda}(x) v = - \mu^2 v, \quad \quad q_{\lambda} = - \lambda^2 f,
\end{equation}
with boundary conditions
\begin{equation} \label{BC4-2D}
  v(0) = 0, \quad v(1) = 0.
\end{equation}

Note that the equation (\ref{Eq4-2D}) is nothing but equation (\ref{Eq2-2D}) in which the parameter $-m^2$ is written as $-\mu^2$ and is interpreted as the spectral parameter of the equation. On the other hand, the boundary conditions (\ref{BC1-2D}) have been replaced by Dirichlet boundary conditions at $x = 0$ and $x=1$. Since the potential $q_{\lambda}$ lies in  $L^1([0,1])$ and is real, the boundary value problem (\ref{Eq4-2D}) - (\ref{BC4-2D}) lies in the framework recalled in Section \ref{Preliminary}. Hence, for all $\mu \in \C$, we can define the fundamental systems of solutions
$$
  \{ c_0(x,\mu^2), s_0(x,\mu^2)\}, \quad \{ c_1(x,\mu^2), s_1(x,\mu^2)\},
$$
of (\ref{Eq4-2D}) by imposing the Cauchy conditions
\begin{equation} \label{FSS-2D}
  \left\{ \begin{array}{cccc} c_0(0,\mu^2) = 1, & c_0'(0,\mu^2) = 0, & s_0(0,\mu^2) = 0, & s_0'(0,\mu^2) = 1, \\
 	                  c_1(1,\mu^2) = 1, & c'_1(1,\mu^2) = 0, & s_1(1,\mu^2) = 0, & s'_1(1,\mu^2) = 1. \end{array} \right.
\end{equation}
Note that the dependence of the FSS on $\lambda^2$ is not written for clarity but implicit.

We then define the characteristic function of (\ref{Eq4-2D}) with boundary conditions (\ref{BC4-2D}) by
\begin{equation} \label{Char-2D}
  \Delta(\mu^2) = W(s_0, s_1),
\end{equation}
whereas the Weyl-Titchmarsh functions are defined by
\begin{equation} \label{WT-2D}
  M(\mu^2) = - \frac{W(c_0, s_1)}{\Delta(\mu^2)} = - \frac{D(\mu^2)}{\Delta(\mu^2)}, \quad \quad N(\mu^2) = -\frac{W(c_1, s_0)}{\Delta(\mu^2)} =  -\frac{E(\mu^2)}{\Delta(\mu^2)}.
\end{equation}
We refer to Section \ref{Preliminary} for the notations.

\begin{rem} \label{M-WellDefined-2D}
1. In the case $\lambda^2 = 0$, we get immediately the following explicit formulas for the characteristic and WT functions.
\begin{equation} \label{Lambda=0}
  \begin{array}{l}
	  c_0(x,\mu^2) = \cosh(\mu x), \quad s_0(x,\mu^2) = \frac{\sinh(\mu x)}{\mu}, \\
		c_1(x,\mu^2) = \cosh(\mu (1-x)), \quad s_1(x,\mu^2) = -\frac{\sinh(\mu (1-x))}{\mu}, \\
		\Delta(\mu^2) = \frac{\sinh(\mu)}{\mu}, \quad D(\mu^2) = E(\mu^2) = \cosh(\mu), \\
		M(\mu^2) = -\mu \coth(\mu), \quad N(\mu^2) = -\mu \coth(\mu).
	\end{array}
\end{equation}	
2. The WT function $M(\mu^2)$ is meromorphic on $\C$ and has an infinite and discrete set of poles $\{ \alpha^2_k\}_{k \geq 1}$ corresponding to "minus" the Dirichlet eigenvalues of $-\frac{d}{dx^2} + q_\lambda$ or equivalently, corresponding to the zeros of the characteristic function $\Delta(\mu^2)$. Let us show that the integers $m^2, \ m \in \Z$ cannot be poles of $M(\mu^2)$ under our general assumption. Assume the converse, \textit{i.e.} there exists $m \in \Z$ such that $m^2$ is a pole of $M(\mu^2)$. Thus $-m^2$ is a Dirichlet eigenvalue of $-\frac{d}{dx^2} + q_\lambda$ and we denote by $u_m(x)$ the associated eigenfunction solution of (\ref{Eq2-2D}). Then the function $u(x,y) = u_m(x) Y_m(y)$ is a nontrivial solution of (\ref{Eq1-2D}) with Dirichlet boundary conditions. We conclude that $u$ is an eigenfunction of $-\Delta_g$ associated to the Dirichlet eigenvalue $\lambda^2$. But this case has been ruled out from the beginning since we assume that $\lambda^2$ is not a Dirichlet eigenvalue of $\Delta_g$. Whence the contradiction.

In particular, we see that $\,0$ cannot be a pole of $M(\mu^2)$ or in other words, that the characteristic function $\Delta(\mu^2)$ doesn't vanish at $0$. Using Corollary \ref{Hadamard}, we have thus the following factorization for $\Delta(\mu^2)$ that will be used later.
\begin{equation} \label{Hadamard-Delta}
  \Delta(\mu^2) = \Delta(0) \prod_{k = 1}^\infty \left( 1 - \frac{\mu^2}{\alpha_{k}^2} \right).
\end{equation}
\end{rem}

We now come back to the expression (\ref{DN-Partiel-1-2D}) of the partial DN map $\Lambda^{m}_g(\lambda^2)$. For all $m \in \Z$, we need to express $u_{m}'(0)$ and $u_{m}'(1)$ in terms of $\psi^0_{m}$ and $\psi^1_{m}$ in order to find the expressions of the coefficients $L^{m}(\lambda^2), T^{m}_R(\lambda^2), T^{m}_L(\lambda^2)$ and $R^{m}(\lambda^2)$. But the solution $u_{m}$ of (\ref{Eq2-2D}) can be written as linear combinations of the FSS $\{ c_0(x,m^2), s_0(x,m^2)\}$ and $\{ c_1(x,m^2), s_1(x,m^2)\}$. Precisely, we have
$$
  u_{m}(x) = \alpha \,c_0(x,m^2) + \beta \,s_0(x,m^2) = \gamma \,c_1(x,m^2) + \delta \,s_1(x,m^2),
$$
for some constants $\alpha,\beta,\gamma,\delta$. Using (\ref{Eq2-2D}) and (\ref{FSS-2D}), we first get
\begin{equation} \label{a1-2D}
  \left( \begin{array}{c} u_{m}(0) \\ u_{m}(1) \end{array} \right) = \left( \begin{array}{c} \psi^0_{m} \\ \psi^1_{m} \end{array} \right) = \left( \begin{array}{c} \alpha \\ \gamma \end{array} \right) = \left( \begin{array}{c} \gamma \,c_1(0,m^2) + \delta \, s_1(0,m^2),  \\ \alpha \,c_0(1,m^2) + \beta \,s_0(1,m^2) \end{array} \right).
\end{equation}
From (\ref{a1-2D}), we obtain in particular
\begin{equation} \label{a2-2D}
  \left( \begin{array}{c} \beta \\ \delta  \end{array} \right) = \left( \begin{array}{cc} -\frac{c_0(1,m^2)}{s_0(1,m^2)} & \frac{1}{s_0(1,m^2)} \\ \frac{1}{s_1(0,m^2)} & -\frac{c_1(0,m^2)}{s_1(0,m^2)} \end{array} \right) \left( \begin{array}{c} \psi^0_{m} \\ \psi^1_{m}  \end{array} \right).
\end{equation}
Also using (\ref{FSS-2D}) and (\ref{a2-2D}), we have
\begin{equation} \label{a3-2D}
  \left( \begin{array}{c} u_{m}'(0) \\ u_{m}'(1)  \end{array} \right) = \left( \begin{array}{c} \beta \\ \delta  \end{array} \right) = \left( \begin{array}{cc} -\frac{c_0(1,m^2)}{s_0(1,m^2)} & \frac{1}{s_0(1,m^2)} \\ \frac{1}{s_1(0,m^2)} & -\frac{c_1(0,m^2)}{s_1(0,m^2)} \end{array} \right) \left( \begin{array}{c} \psi^0_{m} \\ \psi^1_{m}  \end{array} \right).
\end{equation}
Therefore by (\ref{DN-Partiel-1-2D}) and (\ref{a1-2D}) - (\ref{a3-2D}), we obtain.
\begin{equation} \label{DN-Partiel-3-2D}
  \Lambda^{m}_g(\lambda^2) \left( \begin{array}{c} \psi_{m}^0 \\ \psi_{m}^1  \end{array} \right) = \left( \begin{array}{cc} \frac{1}{\sqrt{f(0)}} \frac{c_0(1,m^2)}{s_0(1,m^2)} &  - \frac{1}{\sqrt{f(0)}} \frac{1}{s_0(1,m^2)} \\ \frac{1}{\sqrt{f(1)}} \frac{1}{s_1(0,m^2)} &   - \frac{1}{\sqrt{f(1)}} \frac{c_1(0,m^2)}{s_1(0,m^2)} \end{array} \right) \left( \begin{array}{c} \psi_{m}^0 \\ \psi_{m}^1  \end{array} \right).
\end{equation}
Equivalently, we obtain the following expression for the partial DN map $\Lambda^{m}_g(\lambda^2)$.
\begin{equation} \label{DN-Partiel-4-2D}
  \Lambda^{m}_g(\lambda^2) = \left( \begin{array}{cc} L^{m}(\lambda^2) & T^{m}_R(\lambda^2) \\ T^{m}_L(\lambda^2) & R^{m}(\lambda^2) \end{array} \right) = \left( \begin{array}{cc} \frac{1}{\sqrt{f(0)}} \frac{c_0(1,m^2)}{s_0(1,m^2)} &  - \frac{1}{\sqrt{f(0)}} \frac{1}{s_0(1,m^2)} \\ \frac{1}{\sqrt{f(1)}} \frac{1}{s_1(0,m^2)} &   - \frac{1}{\sqrt{f(1)}} \frac{c_1(0,m^2)}{s_1(0,m^2)} \end{array} \right).
\end{equation}
Finally, using (\ref{FSS-2D}) again, we easily show that
\begin{equation}\label{formule}
  \Delta(m^2) = s_0(1,m^2) = -s_1(0,m^2), \quad M(m^2) = - \frac{c_0(1,m^2)}{s_0(1,m^2)}, \quad N(m^2) =  \frac{c_1(0,m^2)}{s_1(0,m^2)}.
\end{equation}
Therefore
\begin{equation} \label{DN-Partiel-5-2D}
  \Lambda^{m}_g(\lambda^2) = \left( \begin{array}{cc} L^{m}(\lambda^2) & T^{m}_R(\lambda^2) \\ T^{m}_L(\lambda^2) & R^{m}(\lambda^2) \end{array} \right) = \left( \begin{array}{cc} - \frac{1}{\sqrt{f(0)}} \, M(m^2) &  - \frac{1}{\sqrt{f(0)}} \, \frac{1}{\Delta(m^2)} \\ -\frac{1}{\sqrt{f(1)}} \, \frac{1}{\Delta(m^2)} &  -\frac{1}{\sqrt{f(1)}} \, N(m^2) \end{array} \right).
\end{equation}
We emphasize that the quantities $\Delta(m^2)$ and $M(m^2)$ are well defined for all $m \in \Z$ thanks to our assumption on $\lambda^2$, \textit{i.e.} $\lambda^2$ is not a Dirichlet eigenvalue of $\Delta_g$ (see Remark \ref{M-WellDefined-2D} for this point).

The result above shows that the coefficients of $\Lambda^{m}_g(\lambda^2)$ which correspond to the global DN map $\Lambda_g(\lambda^2)$ restricted to the Fourier mode $< Y_m >$, can be simply interpreted in terms of the characteristic and Weyl-Titchmarsh functions $\Delta(m^2),\, M(m^2)$ and $N(m^2)$ associated to the simple ODE (\ref{Eq4-2D}) with boundary data (\ref{BC4-2D}). In particular, the partial DN maps $L(\lambda^2)$ and $R(\lambda^2)$ which correspond to the global DN map with data restricted to $\Gamma_0$ and $\Gamma_1$ respectively, only depend on the Weyl-Titchmarsh functions $M(m^2)$ and $N(m^2)$ respectively (modulo some explicit constants). In fact, it is clear in our model that the knowledge of $L(\lambda^2)$ and $R(\lambda^2)$ is equivalent to the knowledge of the sequence of Weyl-Titchmarsh functions $\{M(m^2)\}_{m \in \Z}$ and $\{N(m^2)\}_{m \in \Z}$ respectively (modulo some constants).

Similarly, the knowledge of the partial DN maps $T_L(\lambda^2)$ and $T_R(\lambda^2)$ which correspond to the global DN map whose data are measured on the disjoint sets $\Gamma_0$ and $\Gamma_1$, is equivalent to the knowledge of the characteristic functions $\{\Delta(m^2)\}_{m \in \Z}$ (modulo some explicit constants). We emphasize that this is the key point that explains our non-uniqueness result for the Calderon problem with data measured on disjoint sets belonging to two distinct connected components of the boundary. Indeed, it is well known that the characteristic function $\Delta(\mu^2)$ does not contain enough information to determine uniquely the potential $q_\lambda$ in (\ref{Eq4-2D}) and thus the metric $g$.


\subsection{The 2D anisotropic Calderon problem with data measured on the same connected component.} \label{Calderon-R-2D}

In this Section, we solve the Calderon problem from the knowledge of  the partial DN map $\Lambda_{g, \Gamma_D, \Gamma_N}$, where $\Gamma_D$ and $\Gamma_N$ are any open subsets living in the same connected component of the boundary, (i.e, we assume that $\Gamma_D, \ \Gamma_N$ are subsets of, either $\{0\}\times T^1$, or $\{1\}\times T^1$). Precisely, we prove the following result.

\begin{thm} \label{MainThm-R-2D}
  Let $(\S,g)$ and $(\S,\tilde{g})$ two Riemannian surfaces of the form (\ref{Metric-2D}), \textit{i.e.}
	$$
	  g = f(x) [dx^2 + dy^2], \quad \tilde{g} = \tilde{f}(x) [dx^2 +  dy^2].
	$$
	We shall add the subscript $\ \tilde{}$ to all the quantities referring to $(\S,\tilde{g})$. Let the frequency $\lambda^2$ be fixed and let
    $\Gamma_D$, $\Gamma_N$ be subsets belonging to the same connected component of $\partial \S$. We assume that
	$$
    \Lambda_{g, \Gamma_D, \Gamma_N} = \Lambda_{\tilde{g}, \Gamma_D, \Gamma_N}.
	$$
    Then \\
	1) If $\lambda^2 = 0$, there exists a function $c > 0$ with $c(0) = 1$ if $\Gamma_D, \ \Gamma_N \subset \{0\}\times T^1$, and $c(1) = 1$ if $\Gamma_D, \ \Gamma_N \subset \{1\}\times T^1$ such that
	$$
	  \tilde{g} = c g.
	$$
	2) If $\lambda^2 \ne 0$, we have $\tilde{g} = g$.
\end{thm}

\begin{rem}
In the case $\lambda^2 = 0$ and when $\Gamma_D$ and $\Gamma_N$ are not disjoint, this result is well known and has been proved first for any smooth Riemannian surface by Lassas and Uhlmann in \cite{LaU}. We also refer to the work of Guillarmou and Tzou \cite{GT1} for similar results and to \cite{GT2} for a recent survey on 2D inverse Calderon problem. 
\end{rem}

\begin{proof}
Assume for instance that  $\Gamma_D, \ \Gamma_N \subset \{0\}\times T^1$. As usually, we identify $\{0\}\times T^1$ with $[-\pi,\pi]$. Without loss of generality, we can always assume that $0 \in \Gamma_D$ and $\pm \pi \in \Gamma_N$. The general case requires minor modifications. It follows from the hypothesis that, for all $\Psi \in C_0^{\infty}(-R,R)$, $R$ small enough, ($R<\pi$), and for all $y \in \Gamma_N$,
\begin{equation}\label{egalitepartial}
\sum_{m \in \Z} \frac{1}{\sqrt{f(0)}} M(m^2) \hat{\Psi}(m) \ e^{imy} = \sum_{m \in \Z} \frac{1}{\sqrt{\tilde{f}(0)}} \tilde{M}(m^2) \hat{\Psi}(m) \ e^{imy},
\end{equation}
where
\begin{equation}\label{coefFourierPsi}
\hat{\Psi}(m) = \frac{1}{2\pi} \ \int_{-\pi}^{\pi} \Psi(y) \ e^{-imy} \ dy
\end{equation}
is the $m^{th}$ Fourier coefficient of $\Psi$. We set
\begin{equation}
\Phi(y) = \sum_{m \in \Z} \left( \frac{1}{\sqrt{f(0)}} M(m^2) -   \frac{1}{\sqrt{\tilde{f}(0)}} \tilde{M}(m^2) \right)    \hat{\Psi}(m) \ e^{imy}.
\end{equation}
From (\ref{egalitepartial}), we deduce that $supp\ \Phi \subset (-R',R')$ for some $R'<\pi$, and clearly, the $m^{th}$ Fourier coefficient of $\Phi$ is
\begin{equation}\label{coefFourierPhi}
\hat{\Phi}(m) =  \left( \frac{1}{\sqrt{f(0)}} M(m^2) -   \frac{1}{\sqrt{\tilde{f}(0)}} \tilde{M}(m^2) \right)   \hat{\Psi}(m).
\end{equation}
By the Paley-Wiener Theorem (\cite{Boa}, Theorem 6.8.1), $\hat{\Psi}(m)$, (resp. $\hat{\Phi}(m)$) can be extended as an entire function, $\mu \rightarrow \hat{\Psi}(\mu)$, (resp. $\mu \rightarrow \hat{\Phi}(\mu)$) of order $1$, of type strictly less than $\pi$. More precisely, there exists $A>0$ and $B <\pi$ such that
\begin{equation}\label{orderpi}
\mid \hat{\Psi}(\mu) \mid \leq A \ e^{B \mid \mu \mid} \ , \  \mid \hat{\Phi}(\mu) \mid \leq A \ e^{B \mid \mu \mid} \ ,\ \forall \mu \in \C.
\end{equation}
Now, we can establish the following elementary lemma:

\begin{lemma}\label{prolongement}
For all $\mu \in \C$,
\begin{equation*}
F(\mu):=\Delta(\mu^2) \tilde{\Delta}(\mu^2) \hat{\Phi}(\mu) + \left( \frac{1}{\sqrt{f(0)}} D(\mu^2) \tilde{\Delta}(\mu^2)
- \frac{1}{\sqrt{\tilde{f}(0)}} \tilde{D}(\mu^2) \Delta(\mu^2)  \right) \hat{\Psi}(\mu) =0.
\end{equation*}
\end{lemma}

\begin{proof}
Using (\ref{orderpi}) and Corollary \ref{OrderHalf}, we see that $F(\mu)$ is an entire function of order $1$
satisfying the following estimate on the imaginary axis :
\begin{equation}
\mid F(iy) \mid \leq C \ e^{B \mid  y\mid}, \ \forall y \in \R,
\end{equation}
for some $C>0$. Moreover, since $M(m^2) = - \frac{D(m^2)}{\Delta(m^2)}$, we  deduce from (\ref{coefFourierPhi}) that $F(m)=0, \ \forall m \in \Z$. Since $B<\pi$, using Carlson's theorem, (see \cite{Boa}, Theorem 9.2.1), we have $F(\mu)=0$ for all $\mu \in \C$.
\end{proof}

\vspace{0.5cm} \noindent
We can deduce from this lemma:

\begin{prop} \label{equalityWT}
Assume that $\Lambda_{g, \Gamma_D, \Gamma_N} = \Lambda_{\tilde{g}, \Gamma_D, \Gamma_N}$,  with $\Gamma_D, \ \Gamma_N \subset \{0\}\times T^1$. Then :
\par\noindent
(1) The zeros of $\Delta(\mu^2)$ and $\tilde{\Delta}(\mu^2)$ coincide, and in particular, for all $\mu \in \C$,
\begin{equation}\label{caractegal}
\Delta (\mu^2) = \tilde{\Delta}(\mu^2).
\end{equation}
(2) For all $\mu \in \C \backslash \R$,
\begin{equation} \label{WTegal}
M (\mu^2)- M(0) = \tilde{M}(\mu^2) - \tilde{M}(0)
\end{equation}
\end{prop}

\begin{proof}

First, we recall that the zeros  of $\Delta (\mu^2)$, (resp. $\tilde{\Delta}(\mu^2)$) are simple (see for instance Theorem 2, p. 30 of \cite{PT}). Second, let us show that the zeros  of $\Delta (\mu^2)$  and $\tilde{\Delta}(\mu^2)$) coincide. For instance, assume that $\Delta(\alpha_k^2) =0, \ k \geq 1$. Taking in Lemma \ref{prolongement}, $\Psi \in C_0^{\infty}(-R,R)$ such that $\hat{\Psi}(\alpha_k)=1$, we obtain
$$
  D(\alpha_k^2) \tilde{\Delta}(\alpha_k^2)=0.
$$	
But, by definition, $D(\mu^2)$ and $\Delta(\mu^2)$ cannot vanish at the same time. We deduce then that $\tilde{\Delta}(\alpha_k^2)=0$.

\vspace{0.2cm}
Hence we infer from Corollary \ref{Hadamard} and (\ref{Hadamard-Delta}) that
\begin{equation}
\Delta(\mu^2) = \frac{\Delta(0)}{\tilde{\Delta}(0)} \ \tilde{\Delta}(\mu^2) \ ,\forall \mu \in \C.
\end{equation}
But we know from Corollary \ref{OrderHalf} that $\Delta(\mu^2) \sim \tilde{\Delta}(\mu^2)$ when $\mu \to +\infty$. We deduce that $\Delta(0) = \tilde{\Delta}(0)$ and then $ \Delta(\mu^2) = \tilde{\Delta}(\mu^2)$ for all $\mu \in \C$. This equality will be useful in the next section.

\vspace{0.5cm}
Now, let us prove the assertion $(2)$ of Proposition \ref{equalityWT}. First, using Lemma \ref{prolongement}, we obtain: for all $\mu \in \C \backslash \R$,
\begin{equation}\label{simplifiee}
\hat{\Phi}(\mu) + \left( \frac{1}{\sqrt{f(0)}} M(\mu^2)
- \frac{1}{\sqrt{\tilde{f}(0)}} \tilde{M}(\mu^2)   \right) \hat{\Psi}(\mu) =0
\end{equation}
We set $m(\mu)=M(\mu^2)$, (resp. $\tilde{m}(\mu) =\tilde{M}(\mu^2))$. Clearly, $m$, (resp. $\tilde{m}$) is a meromorphic function with simple poles $\pm \alpha_k$, (we recall that $\alpha_k \not=0$). Moreover, $Res(m; \ -\alpha_k)= -Res(m;\ \alpha_k)$, (resp. $Res(\tilde{m}; \ -\alpha_k)= -Res(\tilde{m};\ \alpha_k)$.
Now, in (\ref{simplifiee}), taking again a function $\Psi$ such that $\hat{\Psi}(\alpha_k)=1$ or $\hat{\Psi}(-\alpha_k)=1$,  and since $\hat{\Phi}(\mu)$ is entire, we see that the residues at $\pm \alpha_k$ of $\frac{1}{\sqrt{f(0)}} m(\mu)$ and $\frac{1}{\sqrt{\tilde{f}(0)}}\tilde{m}(\mu)$ must be the same. So, we have proved that:
\begin{equation}\label{residuegaux}
\frac{1}{\sqrt{f(0)}}  Res(m; \pm \alpha_k) = \frac{1}{\sqrt{\tilde{f}(0)}}  Res(\tilde{m}; \pm \alpha_k).
\end{equation}


\vspace{0.3cm}\noindent
Secondly, we set
\begin{equation}
F(\mu) =\frac{m(\mu)-m(0)}{\mu}.
\end{equation}
Clearly, $F(\mu)$ is a meromorphic function with simple poles at the zeros $\pm \alpha_k$, $k \geq 1$. Moreover, using the asymptotics of $D(\mu^2)$ and $\Delta(\mu^2)$ given in Corollary \ref{OrderHalf}, a standard calculation (see for instance, \cite{SP}, Chapter 7, p. 227) shows there exists a constant $C>0$ such that for all $N \in \N$:
\begin{equation}\label{carre}
\mid F(z) \mid \ \leq \ C, \ \forall z \in C_N,
\end{equation}
where $C_N$ is the square with vertices $z= (\pm 1\pm i) (N+\half)\pi$. We deduce from (\ref{carre}) that
\begin{equation}
\lim_{N \rightarrow + \infty} \int_{C_N} \frac{F(z)}{z(z-\mu)} \ dz = 0,
\end{equation}
where $\mu \not= 0, \ \pm \alpha_k$ is fixed. Then, using the Cauchy's residue theorem, we obtain the well-known Mittag-Leffler's expansion (see \cite{SP}, Chapter 7, p. 210 for details):
\begin{equation}\label{ML1}
F(\mu) = F(0) + \sum_{k =1}^{+\infty} Res (F;\  \alpha_k) \left( \frac{1}{\mu -\alpha_k} + \frac{1}{\alpha_k} \right) +
\sum_{k =1}^{+\infty} Res (F;\ -\alpha_k) \left( \frac{1}{\mu +\alpha_k} - \frac{1}{\alpha_k} \right).
\end{equation}


\noindent
Clearly, $F(0)=0$ since $m(\mu)$ is an even function of $\mu$ and we also have:
\begin{equation}
Res (F; \pm\alpha_k) = \pm \frac{Res(m; \pm \alpha_k)}{\alpha_k}.
\end{equation}
So, we can rewrite (\ref{ML1}) as
\begin{equation}
M(\mu^2) - M(0) =2  \sum_{k=1}^{+\infty} \frac{Res(m; \alpha_k)}{\alpha_k} \frac{\mu^2}{\mu^2 -\alpha_k^2}.
\end{equation}
Thus, with the help of (\ref{residuegaux}), we obtain:
\begin{equation}
\frac{1}{\sqrt{f(0)}} \left( M(\mu^2)-M(0) \right) = \frac{1}{\sqrt{\tilde{f}(0)}} \left( \tilde{M}(\mu^2)-\tilde{M}(0) \right).
\end{equation}
Using the asymptotics of the Weyl-Titchmarsh function $M(\mu^2)$ as $\mu \to \infty$ given in Corollary \ref{OrderHalf}, we get $f(0)= \tilde{f}(0)$ and the Proposition is proved.
\end{proof}


\vspace{0.5cm}
Now, we can finish the proof of Theorem \ref{MainThm-R-2D}. First, let us examine the obvious case $\lambda^2=0$, (which can  probably be done differently).
When $\lambda^2=0$, the WT functions $M(\mu^2)$, (resp. $\tilde{M}(\mu^2)$) do not depend on the metric, (see (\ref{Lambda=0})). So the unique condition we can get from the hypothesis $\Lambda_{g, \Gamma_D, \Gamma_N} = \Lambda_{\tilde{g}, \Gamma_D, \Gamma_N}$ is  $f(0)= \tilde{f}(0)$, which was obtained above. Thus, we have:
$$
  \tilde{g} = c g, \quad c(x) = \frac{\tilde{f}(x)}{f(x)},
$$
with
$$
c(0) = 1.
$$

\vspace{0.3cm}
Now, let us study the more interesting case $\lambda^2\not=0$. Recalling that the WT functions $M(\mu^2)$ and $\tilde{M}(\mu^2)$ are associated to equation (\ref{Eq4-2D}) with potentials $q_\lambda$ and $\tilde{q_\lambda}$, we thus conclude from (\ref{WTegal}) and the Borg-Marchenko Theorem \ref{BM} that
\begin{equation} \label{d6}
  q_\lambda(x) = \tilde{q_\lambda}(x), \quad \forall x \in [0,1].
\end{equation}	
In other words, $f(x) = \tilde{f}(x), \ \forall x \in [0,1]$ and thus $\tilde{g} = g$. This proves the result.
\end{proof}


\subsection{The 2D anisotropic Calderon problem with data measured on distinct connected components.} \label{Calderon-T-2D}

In this Section, we prove our first main Theorem, namely we give a counterexample to uniqueness for the anisotropic Calderon problem, at a nonzero frequency, for Riemannian surfaces with partial data measured on sets belonging to two distinct connected components. Precisely, we prove:

\begin{thm} \label{MainThm-T-2D}
Let $(\S,g)$ and $(\S,\tilde{g})$ two Riemannian surfaces of the form (\ref{Metric-2D}), \textit{i.e.}
$$
	g = f(x) [dx^2 + dy^2], \quad \tilde{g} = \tilde{f}(x) [dx^2 +  dy^2].
$$
We shall add the subscript $\ \tilde{}$ to all the quantities referring to $(\S,\tilde{g})$. Let the frequency $\lambda^2$ be fixed and let
$\Gamma_D$, $\Gamma_N$ be subsets belonging to distinct connected components of $\partial \S$. Then \\

\noindent 1) If $\lambda^2 = 0$ and $\Lambda_{g, \Gamma_D, \Gamma_N} (0)= \Lambda_{\tilde{g}, \Gamma_D, \Gamma_N}(0)$, there exists a function $c > 0$ with $c(0) = 1$ if $\Gamma_D  \subset \{1\}\times T^1$, and $\Gamma_N \subset \{0\}\times T^1$, (resp.  $c(1) = 1$ if $\Gamma_D \subset \{0\}\times T^1$ and $\Gamma_N \subset \{1\}\times T^1$)  such that
$$
	\tilde{g} = c g.
$$

\noindent 2) If $\lambda^2 \ne 0$, there exists an explicit infinite dimensional family of metrics $\tilde{g} = cg$ with $c>0$ and $c(0) = c(1) = 1$ that satisfies
$$
  \Lambda_{g, \Gamma_D, \Gamma_N}(\lambda^2) = \Lambda_{\tilde{g}, \Gamma_D, \Gamma_N}(\lambda^2).
$$	
\end{thm}

\begin{proof}
Assume for instance that  $\Gamma_D \subset \{1\}\times T^1$ and $\Gamma_N \subset \{0\}\times T^1$. We follow the same strategy as in the previous section and we use  the same notation. We assume that $0 \in \Gamma_D$ and $\pm \pi \in \Gamma_N$ where $T^1$ is identified with $[-\pi,\pi]$.  Thus, for all $\Psi \in C_0^{\infty}(-R,R)$, $R$ small enough, ($R<\pi$), and for all $y \in \Gamma_N$,
\begin{equation}\label{egalitepartial2}
\sum_{m \in \Z} \frac{1}{\sqrt{f(0)}} \frac{1}{\Delta(m^2)} \hat{\Psi}(m) \ e^{imy} = \sum_{m \in \Z} \frac{1}{\sqrt{\tilde{f}(0)}} \frac{1}{\tilde{\Delta}(m^2)} \hat{\Psi}(m) \ e^{imy}.
\end{equation}
Then, repeating exactly the argument of the previous section, we see that
\begin{equation} \label{e2}
  f(0) = \tilde{f}(0),
\end{equation}
and
\begin{equation} \label{e3}
  \alpha_{k}^2 = \tilde{\alpha}_{k}^2, \quad \forall k \geq 1,
\end{equation}
where $(\alpha_{k}^2)_{k \geq 1}$ and $(\tilde{\alpha}_{k}^2)_{k \geq 1}$ are the zeros of $\Delta(\mu^2)$ and $\tilde{\Delta}(\mu^2)$.

\vspace{0.3cm}
Obviously, in the case $\lambda^2=0$, since the characteristic functions $\Delta(\mu^2)$ and $\tilde{\Delta}(\mu^2)$ do not depend on the metrics (see (\ref{Lambda=0})), the unique condition we can get from the hypothesis $\Lambda_{g, \Gamma_D, \Gamma_N} = \Lambda_{\tilde{g}, \Gamma_D, \Gamma_N}$ is
$f(0) = \tilde{f}(0)$. Hence, we have
$$
  \tilde{g} = c g, \quad c(x) = \frac{\tilde{f}(x)}{f(x)},
$$
with
$$
c(0) = 1.
$$
This proves the assertion $(1)$.

\vspace{0.2cm}
Now, let us examine the case $\lambda^2\not=0$. First, we make the following remark : if the conditions (\ref{e2}) and (\ref{e3}) are satisfied, we can prove easily that
$\Lambda_{g, \Gamma_D, \Gamma_N} = \Lambda_{\tilde{g}, \Gamma_D, \Gamma_N}$. Indeed, from the proof of Proposition \ref{equalityWT} and (\ref{e3}), we deduce that
\begin{equation}
\Delta(\mu^2) = \tilde{\Delta}(\mu^2) \ ,\ \forall \mu \in \C.
\end{equation}
It follows that, under the conditions (\ref{e2}) and (\ref{e3}), we have  $T_R^m(\lambda^2) = \tilde{T}_R^m(\lambda^2)$, for all $m \in \Z$, or equivalently $T_R(\lambda^2)= \tilde{T}_R(\lambda^2)$ which is a stronger condition than $\Lambda_{g, \Gamma_D, \Gamma_N} = \Lambda_{\tilde{g}, \Gamma_D, \Gamma_N}$. To summarize, this shows that the hypothesis $\Lambda_{g, \Gamma_D, \Gamma_N} = \Lambda_{\tilde{g}, \Gamma_D, \Gamma_N}$ is equivalent to both conditions (\ref{e2}) and (\ref{e3}).

\vspace{0.2cm}
We would like to emphasize now that the condition (\ref{e3}) is nothing but a condition of isospectrality of the potentials $q_\lambda$ and $\tilde{q}_\lambda$ appearing in equation (\ref{Eq4-2D}). Indeed, the sequences $(\alpha_{k}^2)_{k \geq 1}$ and $(\tilde{\alpha}_{k}^2)_{k \geq 1}$ correspond to "minus" the Dirichlet spectra of the Schr\"odinger operators $-\frac{d^2}{dx^2} + q_\lambda$ and $-\frac{d^2}{dx^2} + \tilde{q}_\lambda$. Hence, condition (\ref{e3}) means exactly that the Dirichlet spectra of $-\frac{d^2}{dx^2} + q_\lambda$ and $-\frac{d^2}{dx^2} + \tilde{q}_\lambda$ coincide, (we recall that the eigenvalues of these Schr\"odinger operators are always simple, (see for instance \cite{PT}, Theorem 2, p 30 or \cite{Ze}, Theorem 4.3.1)). In other words, the potentials $q_\lambda$ and $\tilde{q}_\lambda$ are isospectral.

\vspace{0.2 cm}
Hence, the hypothesis $\Lambda_{g, \Gamma_D, \Gamma_N} = \Lambda_{\tilde{g}, \Gamma_D, \Gamma_N}$ is equivalent to both conditions (\ref{e2}) and $q_\lambda$, $\tilde{q}_\lambda$ are isospectral.

\vspace{0.2cm}
It turns out that the isospectral class in $L^2([0,1])$ associated to a given potential $q$ and equation (\ref{Eq4-2D}) has been the subject of intensive studies. We refer to the beautiful book \cite{PT} for a clear presentation of the results concerning this problem. For instance, the isospectral class in $L^2([0,1])$ of a given potential $q$ has been shown to be a real analytic submanifold of $L^2([0,1])$ lying in the hyperplane of all functions with mean $[q] = \int_0^1 q(s) ds$ (see Theorem 4.1., p 69 in \cite{PT}).

Even better, we have the following explicit characterization of the isospectral class of a given potential $q \in L^2([0,1])$.

\begin{thm}[\cite{PT}, Thm 5.2., p 102] \label{Char-PT-1}
Let $q \in L^2([0,1])$. Denote by $(v_k)_{k \geq 1} = (v_k(x,q))_{k \geq 1}$ the normalized eigenfunctions of the boundary value problem (\ref{Eq4-2D}) - (\ref{BC4-2D}) with the potential $q$. The eigenfunctions $(v_k)_{k \geq 1}$ are associated to the eigenvalues $(\alpha^2_k)_{k \geq 1}$.

Let $\xi \in l^2_1 = \{ \xi = (\xi_1, \xi_2, \dots), \ \sum_{k \geq 1} k^2 |\xi_k|^2 \ < \infty \ \}$. Define the infinite matrix $\Theta(x,\xi,q) = (\theta_{ij}(x,\xi,q))_{i,j \geq 1}$ by
$$
  \theta_{ij}(x,\xi,q) = \delta_{ij} + (e^{\xi_i} - 1) \int_x^1 v_i(s) v_j(s) ds.
$$			
Define also the determinant of $\Theta$ as the limit of the determinants of its principal minors, \textit{i.e.}
$$
  \det \Theta = \lim_{n \to \infty} \det \Theta^n, \quad \Theta^n = (\theta_{ij})_{1 \leq i,j \leq n}.
$$			
Then the isopectral class of $q$ is completely and explicitly described by the family of potentials
$$
  q_{\xi}(x) = q(x) - 2 \frac{d^2}{dx^2} \log \det \Theta(x,\xi,q),
$$
where $\xi \in l^2_1$. It is implicit in the statement of this result that the determinant of $\Theta$ always exists and never vanishes.
\end{thm}

In particular, to each sequence $\xi^k \in l^2_1$ defined by $\xi^k = (\xi^k_j)_{j \geq 1}$ with $\xi^k_j = t\delta_{kj}$, $t \in \R$, we can find a one parameter family of explicit isospectral potentials to $q \in L^2([0,1])$ by
\begin{equation} \label{Iso1}
  q_{k,t}(x) = q(x) - 2 \frac{d^2}{dx^2} \log \theta_{k,t}(x), \quad \quad \forall t \in \R,
\end{equation} 	
where
\begin{equation} \label{Iso2}
  \theta_{k,t}(x) = 1 + (e^t - 1) \int_x^1 v_k^2(s) ds.
\end{equation} 		
Let us make a few remarks on the family of potentials $(q_{k,t})_{k \geq 1, \ t \in \R}$ under the extra assumption that the potential $q$ is smooth on $[0,1]$. Then,

\begin{itemize}
\item First, the potentials $q_{k,t}$ are smooth on $[0,1]$ for all $k \geq 1$ and for all $t \in \R$. Indeed, the eigenfunctions $v_k(x,q)$ are smooth on $[0,1]$ by elliptic regularity. Hence, the functions $\theta_{k,t}$ are also smooth and never vanish on $[0,1]$ for all $k \geq 1$ and for all $t \in \R$ by (\ref{Iso2}). This proves the smoothness of $q_{k,t}$ thanks to  (\ref{Iso1}).

\item Second, for all $k \geq 1$ and for all $t \in \R$, $q_{k,t}(0) = q(0)$ and $q_{k,t}(1) = q(1)$. Indeed, a short calculation shows that
\begin{equation} \label{f1}
  2 \frac{d^2}{dx^2} \log \theta_{k,t}(x) = -2\frac{(e^t - 1) \left[ 1 + (e^t - 1) \int_x^1 v_k^2(s) ds \right] v_k'(x) v_k(x) + (e^t - 1)^2 v_k^4(x)}{\left[ 1 + (e^t - 1) \int_x^1 v_k^2(s) ds \right]^2}.
\end{equation}
Hence, since $v_k(0) = 0$ and $v_k(1) = 0$ by definition, we get the result using (\ref{Iso1}).

\item Third, if moreover $q > 0$ (resp. $q<0$), then for all $k \geq 1$, there exists $T_k > 0$ such that $q_{k,t} >0$ (resp. $q_{k,t} < 0$) for all $-T_k < t < T_k$. Indeed, from (\ref{f1}), it is clear that for a fixed $k \geq 1$, the function $2 \frac{d^2}{dx^2} \log \theta_{k,t}(x)$ can be made arbitrarily small as $t \to 0$ uniformly w.r.t. $x \in [0,1]$. Whence the result thanks to (\ref{Iso1}).
\end{itemize}

We now come back to our initial problem, that is given a frequency $\lambda^2 \ne 0$ and a smooth positive function $f(x)$, find all the \emph{smooth positive} functions $\tilde{f}$ such that (\ref{e2}) and (\ref{e3}) are satisfied. Define the smooth potential $q_\lambda = -\lambda^2 f$ as in (\ref{Eq4-2D}). The potential $q_\lambda$ is either positive or negative according to the sign of $\lambda^2$. Then, as discussed above, (\ref{e3}) is equivalent to finding the smooth positive or negative potentials $\tilde{q}_\lambda = -\lambda^2 \tilde{f}$ that are isospectral to $q_\lambda$. Thus in order to prove our result and according to Theorem \ref{Char-PT-1}, it suffices to show that the intersection of the isospectral class of $q_\lambda$ in $L^2([0,1])$ with the set of smooth positive or negative potentials satisfying (\ref{e2}) is infinite dimensional.

But using the three remarks above, we see that for all $k \geq 1$, there exists $T_k > 0$ such that the potentials $q_{\lambda, k,t}$ given by (\ref{Iso1}) - (\ref{Iso2}) with $q$ replaced by $q_\lambda$ are smooth positive or negative (according to the sign of $\lambda^2$) on $[0,1]$ for all $-T_k < t < T_k$ and satisfy (\ref{e2}). We thus conclude that the infinite dimensional family of metrics (\ref{Metric-2D}) parametrized by the positive functions
\begin{equation} \label{Iso3}
  f_{\lambda,k,t}(x) = f(x) + \frac{2}{\lambda^2} \frac{d^2}{dx^2} \log \theta_{k,t}(x), \quad \forall k \geq 1, \quad -T_k < t < T_k,
\end{equation} 	
with
\begin{equation} \label{Iso4}
  \theta_{k,t}(x) = 1 + (e^t - 1) \int_x^1 v_k^2(s) ds,
\end{equation} 		
where $v_k$ is the normalized eigenfunction of (\ref{Eq4-2D}) - (\ref{BC4-2D}) associated to the eigenvalues $\alpha^2_k$, has the same partial DN map
 $\Lambda_{g, \Gamma_D, \Gamma_N}$ as the metric (\ref{Metric-2D}) associated to $f$. This finishes the proof of the Theorem. 	
\end{proof}


\subsection{The 2D anisotropic Calderon with a potential}

In this Section, we treat the anisotropic Calderon problem \textbf{(Q3)} with a potential $V \in L^\infty(\S)$ such that $V = V(x)$ for our family of metrics (\ref{Metric-2D}). The global DN map $\Lambda_{g,V}(\lambda^2)$ associated to the Dirichlet problem
\begin{equation} \label{Schrodinger-2D}
  \left\{ \begin{array}{cc} (-\Delta_g + V) u = \lambda^2 u, & \textrm{on} \ \S, \\ u = \psi, & \textrm{on} \ \partial \S, \end{array} \right.
\end{equation}
with $\lambda^2$ not belonging to the Dirichlet spectrum of $-\Delta_g + V$ can be constructed in the same way as in Section \ref{2D} since $V = V(x)$ respects the symmetry of $(\S,g)$. On each Fourier modes $Y_m = e^{imy}$, we get the following expression for the induced DN map:
\begin{equation} \label{DN-Partiel-Potentiel-2D}
  \Lambda^{m}_{g,V}(\lambda^2) = \left( \begin{array}{cc} L_V^{m}(\lambda^2) & T^{m}_{R,V}(\lambda^2) \\ T^{m}_{L,V}(\lambda^2) & R_V^{m}(\lambda^2) \end{array} \right) = \left( \begin{array}{cc} - \frac{1}{\sqrt{f(0)}} \, M_V(m^2) &  - \frac{1}{\sqrt{f(0)}} \, \frac{1}{\Delta_V(m^2)} \\ -\frac{1}{\sqrt{f(1)}} \, \frac{1}{\Delta_V(m^2)} &  -\frac{1}{\sqrt{f(1)}} \, N_V(m^2) \end{array} \right),
\end{equation}
where the characteristic and Weyl-Titchmarsh functions $\Delta_V(m^2), M_V(m^2)$ and $N_V(m^2)$ defined by (\ref{Char-2D})-(\ref{WT-2D}) are associated to the radial ODE
\begin{equation} \label{Eq-Schrodinger-2D}
  -v'' + q_{\lambda,V}(x) v = - \mu^2 v, \quad \quad q_{\lambda,V} = (V - \lambda^2) f,
\end{equation}
with boundary conditions
\begin{equation} \label{BC-Schrodinger-2D}
  v(0) = 0, \quad v(1) = 0.
\end{equation}
We also recall the dictionnary between the above coefficients of the DN map and the notations used in the Introduction
$$
  L_V(\lambda^2) = \Lambda_{g,V,\Gamma_0}(\lambda^2), \quad R_V(\lambda^2) = \Lambda_{g,V,\Gamma_1}(\lambda^2),
$$
$$	
	T_{L,V}(\lambda^2) = \Lambda_{g,V,\Gamma_0, \Gamma_1}(\lambda^2), \quad T_{R,V}(\lambda^2) = \Lambda_{g,V,\Gamma_1, \Gamma_0}(\lambda^2),
$$
where $\Gamma_0 = \{0\} \times T^1$ and $\Gamma_1 = \{1\} \times T^1$.

We prove

\begin{thm} \label{MainThm-Schrodinger-2D}
Let $(\S,g)$ a Riemannian surface of the form (\ref{Metric-2D}), \textit{i.e.}
$$
	g = f(x) [dx^2 + dy^2].
$$
Let $V, \tilde{V} \in L^\infty(\S)$ be two potentials that only depend on the variable $x$. Let the frequency $\lambda^2$ be fixed and not belonging to the Dirichlet spectra of $-\Delta_g + V$ and $-\Delta_g + \tilde{V}$. Let $\Gamma_D$, $\Gamma_N$ be nonempty open subsets belonging to $\partial \S$. Then \\

\noindent 1) If $\Gamma_D, \Gamma_N$ belong to the same connected component of $\partial \S$ and $\Lambda_{g, V, \Gamma_D, \Gamma_N}(\lambda^2) = \Lambda_{g, \tilde{V}, \Gamma_D, \Gamma_N}(\lambda^2)$, then
$$
	\tilde{V} = V.
$$

\noindent 2) If $\Gamma_D, \Gamma_N$ belong to distinct connected components of $\partial \S$, then there exists an explicit infinite dimensional family of potentials $\tilde{V}$ that satisfies
$$
  \Lambda_{g, V, \Gamma_D, \Gamma_N}(\lambda^2) = \Lambda_{g, \tilde{V}, \Gamma_D, \Gamma_N}(\lambda^2).
$$	
\end{thm}

\begin{proof}
1) Assume for instance that $\Gamma_D, \Gamma_N \subset \Gamma_0$ and $\Lambda_{g, V, \Gamma_D, \Gamma_N}(\lambda^2) = \Lambda_{g, \tilde{V}, \Gamma_D, \Gamma_N}(\lambda^2)$. Then the same proof as in Theorem \ref{MainThm-R-2D} shows first that
\begin{equation} \label{WT=2D}
  M_{V}(\mu^2) - M_{V}(0)= M_{\tilde{V}}(\mu^2) -M_{\tilde{V}}(0), \quad \forall \mu \in \C \setminus \R.
\end{equation}
Hence the Borg-Marchenko Theorem \ref{BM} entails
\begin{equation} \label{QV=2D}
  q_{\lambda,V} = q_{\lambda, \tilde{V}}, \quad \textrm{on} \ [0,1].
\end{equation}
We conclude using (\ref{QV=2D}) and (\ref{Eq-Schrodinger-2D}) that $V = \tilde{V}$ on $[0,1]$. \\

\noindent 2) Assume that $\Gamma_D$ and $\Gamma_N$ belong to distinct connected components of $\partial \S$ and that
$$
  \Lambda_{g, V, \Gamma_D, \Gamma_N}(\lambda^2) = \Lambda_{g, \tilde{V}, \Gamma_D, \Gamma_N}(\lambda^2).
$$
Then the same proof as in Theorem \ref{MainThm-T-2D} shows first that our hypothesis is equivalent to the unique condition
\begin{equation} \label{Char=2D}
  \Delta_{V}(\mu^2) = \Delta_{\tilde{V}}(\mu^2), \quad \forall \mu \in \C.
\end{equation}
But, as explained in the proof of Theorem \ref{MainThm-T-2D}, this condition is in turn equivalent to the fact that the potentials $q_{\lambda,V}$ and $q_{\lambda,\tilde{V}}$ are isospectral. We deduce thus from \cite{PT} that given a potential $V \in L^\infty(\S)$ as above, there exists an infinite dimensional family of explicit potentials $\tilde{V}$ satisfying (\ref{Char=2D}). More precisely, the family
\begin{equation} \label{Iso-Schrodinger-2D}
  \tilde{V}_{\lambda,k,t}(x) = V(x) - \frac{2}{f(x)} \frac{d^2}{dx^2} \log \theta_{k,t}(x), \quad \forall k \geq 1, \quad t \in \R,
\end{equation} 	
with
\begin{equation} \label{Iso-Schrodinger1-2D}
  \theta_{k,t}(x) = 1 + (e^t - 1) \int_x^1 v_k^2(s) ds,
\end{equation} 		
where $v_k$ is the normalized eigenfunction of (\ref{Eq-Schrodinger-2D}) - (\ref{BC-Schrodinger-2D}) associated to the eigenvalues $\alpha^2_k$ satisfies
$$
  \Lambda_{g, V, \Gamma_D, \Gamma_N}(\lambda^2) = \Lambda_{g, \tilde{V}_{\lambda,n,t}, \Gamma_D, \Gamma_N}(\lambda^2).
$$
Note here that there is no restriction in the range of the parameter $t \in \R$ since we work in the class of potentials $V \in L^{\infty}([0,1])$. In particular, we don't need $\tilde{V}$ to be smooth anylonger.
\end{proof}


\Section{The three dimensional case} \label{3D}

\subsection{The model and the Dirichlet-to-Neumann map} \label{Model-DNmap}

In this Section, we work on a Riemannian manifold $(M,g)$ which has the topology of a cylinder $M = [0,1] \times T^2$ and that is equipped with a Riemannian metric
\begin{equation} \label{Metric}
  g = f(x) dx^2 + f(x) dy^2 + h(x) dz^2,
\end{equation}
where $f, h$ are smooth positive functions on $M$. Here $(x,y,z)$ is a global coordinates system on $M$ and $T^2$ stands for the two dimensional torus. The metric $g$ obviously possesses two commuting Killing vector fields $\partial_y$ and $\partial_z$ that generate a two-parameter Abelian group of isometries on $M$. These isometries will be referred to as the toroidal symmetries of $M$. The boundary $\partial M$ of $M$ is disconnected and consists in the disjoint union of two copies of $T^2$, precisely
$$
  \partial M = \Gamma_0 \cup \Gamma_1, \quad \Gamma_0 = \{0\} \times T^2, \quad \Gamma_1 = \{1\} \times T^2.
$$

In our global coordinate system $(x,y,z)$, the positive Laplace-Beltrami operator has the expression
$$
  -\Delta_g = \frac{1}{f} \left( -\partial_x^2 - \frac{1}{2} (\log h)' \partial_x - \partial_y^2 - \frac{f}{h} \partial_z^2 \right).
$$
We look at the Dirichlet problem at a frequency $\lambda^2$ on $M$ such that $\lambda^2 \notin \{ \lambda_j^2\}_{j \geq 1}$ where $\{ \lambda_j^2\}_{j \geq 1}$ is the Dirichlet spectrum of $-\Delta_g$. We consider thus the solutions $u$ of
\begin{equation} \label{Eq1}
  - \partial_x^2 u - \frac{1}{2} (\log h)' \partial_x u - \partial_y^2 u - \frac{f}{h} \partial_z^2 u - \lambda^2 f u = 0, \quad \textrm{on} \ M,
\end{equation}
together with the boundary condition
\begin{equation} \label{BC1}
  u = \psi \quad \textrm{on} \ \partial M.
\end{equation}

In order to construct the DN map, we shall use the notations introduced in Section \ref{2D}. Recall that the boundary $\partial M$ of $M$ has two disjoint components $\partial M = \Gamma_0 \cup \Gamma_1$ where $\Gamma_0 \simeq \Gamma_1 \simeq T^2$. Hence, we can decompose the Sobolev spaces $H^s(\partial M)$ as $H^s(\partial M) = H^s(\Gamma_0) \oplus H^s(\Gamma_1)$ for any $s \in \R$ and we shall use the vector notation
$$
 \varphi = \left( \begin{array}{c} \varphi^0 \\ \varphi^1 \end{array} \right),
$$
for all elements $\varphi$ of $H^s(\partial M) = H^s(\Gamma_0) \oplus H^s(\Gamma_1)$. The DN map is a linear operator from $H^{1/2}(\partial M)$ to $H^{-1/2}(\partial M)$ and thus has the structure of an operator valued $2 \times 2$ matrix
$$
  \Lambda_g(\lambda^2) = \left( \begin{array}{cc} L(\lambda^2) & T_R(\lambda^2) \\ T_L(\lambda^2) & R(\lambda^2) \end{array} \right),
$$
where $L(\lambda^2), R(\lambda^2), T_R(\lambda^2), T_L(\lambda^2)$ are operators from $H^{1/2}(T^2)$ to $H^{-1/2}(T^2)$. Finally, we recall the dictionary established in Section \ref{2D}
$$
  L(\lambda^2) = \Lambda_{g,\Gamma_0}(\lambda^2), \quad R(\lambda^2) = \Lambda_{g,\Gamma_1}(\lambda^2),
$$
$$	
	T_L(\lambda^2) = \Lambda_{g,\Gamma_0, \Gamma_1}(\lambda^2), \quad T_R(\lambda^2) = \Lambda_{g,\Gamma_1, \Gamma_0}(\lambda^2),
$$
which makes the link with the notations used in the Introduction.

Now we use the toroidal symmetry of $(M,g)$ to find a simple expression of the DN map. We first write $\psi = (\psi^0, \psi^1) \in H^{1/2}(\Gamma_0) \times H^{1/2}(\Gamma_1)$ using their Fourier series as
$$
 \psi^0 = \sum_{m,n \in \Z} \psi^0_{mn} Y_{mn}, \quad \psi^1 = \sum_{m,n \in \Z} \psi^1_{mn} Y_{mn},
$$
where
$$
  Y_{mn}(y,z) = e^{imy} e^{inz}.
$$
Note that for any $s \in \R$, the space $H^{s}(T^2)$ can be described as
$$
  \varphi \in H^{s}(T^2) \ \Longleftrightarrow \ \left\{ \varphi \in \D'(T^2), \ \varphi = \sum_{m,n \in \Z} \varphi_{mn} Y_{mn}, \quad \sum_{m,n \in \Z} (1 + m^2 + n^2)^{s} |\varphi_{mn}|^2 < \infty \ \right\}.
$$
Next, the unique solution $u$ of (\ref{Eq1}) - (\ref{BC1}) can be looked for in the form
$$
  u = \sum_{m,n \in \Z} u_{mn}(x) Y_{mn}(y,z),
$$
and for all $m,n \in \Z$, the function $u_{mn}$ is the unique solution of the ODE (w.r.t. $x$) with boundary conditions
\begin{equation} \label{Eq2}
  \left\{ \begin{array}{c} -u_{mn}'' - \frac{1}{2} (\log h)' u_{mn}' + m^2 u_{mn} + n^2 \frac{f}{h} u_{mn} - \lambda^2 f u_{mn} = 0, \\
    u_{mn}(0) = \psi^0_{mn}, \quad u_{mn}(1) = \psi^1_{mn}. \end{array} \right.
\end{equation}
For later use, we get rid of the term of order $1$ in (\ref{Eq2}) by introducing the new functions
\begin{equation} \label{uTOv}
  v_{mn} = (h)^{1/4} u_{mn},
\end{equation}
which then satisfy
\begin{equation} \label{Eq3-3D}
  \left\{ \begin{array}{c} -v_{mn}'' + \frac{[(\log h)']^2}{16} v_{mn} + \frac{(\log h)''}{4} v_{mn} + m^2 v_{mn} + n^2 \frac{f}{h} v_{mn} - \lambda^2 f v_{mn} = 0, \\  v_{mn}(0) = (h(0))^{1/4} \psi^0_{mn}, \quad v_{mn}(1) = (h(1))^{1/4} \psi^1_{mn}. \end{array} \right.
\end{equation}

The DN map is now diagonalized on the Hilbert basis $\{ Y_{mn} \}_{m,n \in \Z}$ and is shown to have a very simple expression on each harmonic. First, a short calculation using the particular form of the metric (\ref{Metric}) shows that for smooth enough boundary data $\psi$
$$
  \Lambda_g(\lambda^2) \left( \begin{array}{c} \psi^0 \\ \psi^1 \end{array} \right) = \left( \begin{array}{c} \left( \partial_\nu u \right)_{|\Gamma_0} \\ \left( \partial_\nu u \right)_{|\Gamma_1} \end{array} \right) = \left( \begin{array}{c} - \frac{1}{\sqrt{f(0)}} (\partial_x u)_{|x=0} \\ \frac{1}{\sqrt{f(1)}} (\partial_x u)_{|x=1} \end{array} \right).
$$
Hence, if we make the DN map act on the vector space generated by the harmonic $Y_{mn}$, we get
\begin{equation} \label{DN-Partiel-0}
  \Lambda_g(\lambda^2) \left( \begin{array}{c} \psi_{mn}^0 Y_{mn} \\ \psi_{mn}^1 Y_{mn} \end{array} \right) = \left( \begin{array}{c} - \frac{1}{\sqrt{f(0)}} u_{mn}'(0) Y_{mn} \\ \frac{1}{\sqrt{f(1)}} u_{mn}'(1) Y_{mn} \end{array} \right).
\end{equation}
We denote
$$
  \Lambda_g(\lambda^2)_{|<Y_{mn}>} = \Lambda^{mn}_g(\lambda^2) = \left( \begin{array}{cc} L^{mn}(\lambda^2) & T^{mn}_R(\lambda^2) \\ T^{mn}_L(\lambda^2) & R^{mn}(\lambda^2) \end{array} \right),
$$
the restriction of the global DN map to each harmonic $<Y_{mn}>$. This operator has the structure of a $2 \times 2$ matrix and satisfies for all $m,n \in \Z$
\begin{equation} \label{DN-Partiel-1}
  \left( \begin{array}{cc} L^{mn}(\lambda^2) & T^{mn}_R(\lambda^2) \\ T^{mn}_L(\lambda^2) & R^{mn}(\lambda^2) \end{array} \right) \left( \begin{array}{c} \psi_{mn}^0  \\ \psi_{mn}^1  \end{array} \right) = \left( \begin{array}{c} - \frac{1}{\sqrt{f(0)}} u_{mn}'(0) \\ \frac{1}{\sqrt{f(1)}} u_{mn}'(1) \end{array} \right).
\end{equation}
Using the change of functions (\ref{uTOv}), we get the equivalent expression for the action of the partial DN maps $\Lambda^{mn}_g(\lambda^2)$ on vectors $(\psi_{mn}^0, \psi_{mn}^1) \in \C^2$
\begin{equation} \label{DN-Partiel-2}
  \left( \begin{array}{cc} L^{mn}(\lambda^2) & T^{mn}_R(\lambda^2) \\ T^{mn}_L(\lambda^2) & R^{mn}(\lambda^2) \end{array} \right) \left( \begin{array}{c} \psi_{mn}^0  \\ \psi_{mn}^1  \end{array} \right) = \left( \begin{array}{c} -\frac{1}{\sqrt{f(0)}} \left( -\frac{h'(0)}{4 h^{5/4}(0)} v_{mn}(0) + \frac{1}{h^{1/4}(0)} v'_{mn}(0) \right)  \\ \frac{1}{\sqrt{f(1)}} \left( -\frac{h'(1)}{4 h^{5/4}(1)} v_{mn}(1) + \frac{1}{h^{1/4}(1)} v'_{mn}(1) \right) \end{array} \right).
\end{equation}

As in Section \ref{2D}, we can further simplify the partial DN maps $\Lambda^{mn}_g(\lambda^2)$ by interpreting their coefficients as the characteristic and Weyl-Titchmarsh functions of the ODE (\ref{Eq3-3D}) with appropriate boundary conditions. Here is the procedure. First fix $n \in \Z$ and consider the ODE
\begin{equation} \label{Eq4-3D}
  -v'' + q_{\lambda n}(x) v = - \mu^2 v, \quad \quad q_{\lambda n} = \frac{[(\log h)']^2}{16}  + \frac{(\log h)''}{4}  + n^2 \frac{f}{h} - \lambda^2 f,
\end{equation}
with boundary conditions
\begin{equation} \label{BC4-3D}
  v(0) = 0, \quad v(1) = 0.
\end{equation}

Note that the equation (\ref{Eq4-3D}) is nothing but equation (\ref{Eq3-3D}) in which the parameter $-m^2$ is written as $-\mu^2$ and is interpreted as the spectral parameter of the equation. On the other hand, the boundary conditions (\ref{BC1}) have been replaced by Dirichlet boundary conditions at $x = 0$ and $x=1$. Since the potential $q_{\lambda n} \in L^1([0,1])$ and is real, we can define for all $n \in \Z$ and all $\mu \in \C$ the fundamental systems of solutions
$$
  \{ c_0(x,\mu^2,n^2), s_0(x,\mu^2,n^2)\}, \quad \{ c_1(x,\mu^2,n^2), s_1(x,\mu^2,n^2)\},
$$
of (\ref{Eq4-3D}) by imposing the Cauchy conditions
\begin{equation} \label{FSS-3D}
  \left\{ \begin{array}{cccc} c_0(0,\mu^2,n^2) = 1, & c_0'(0,\mu^2,n^2) = 0, & s_0(0,\mu^2,n^2) = 0, & s_0'(0,\mu^2,n^2) = 1, \\
 	                  c_1(1,\mu^2,n^2) = 1, & c'_1(1,\mu^2,n^2) = 0, & s_1(1,\mu^2,n^2) = 0, & s'_1(1,\mu^2,n^2) = 1. \end{array} \right.
\end{equation}
Note that the dependence of the FSS on $\lambda^2$ is not written for clarity but implicit. Also, it is clear that (\ref{Wronskian-FSS}) is satisfied. Finally, the FSS $\{ c_0(x,\mu^2,n^2), s_0(x,\mu^2,n^2)\}$ and $\{ c_1(x,\mu^2,n^2), s_1(x,\mu^2,n^2)\}$ are entire functions with respect to the variable $\mu^2 \in \C$.

Following Section \ref{Preliminary} we thus define the characteristic function of (\ref{Eq4-3D}) with boundary conditions (\ref{BC4-3D}) by
\begin{equation} \label{Char-3D}
  \Delta(\mu^2,n^2) = W(s_0, s_1),
\end{equation}
while we define the Weyl-Titchmarsh functions by
\begin{equation} \label{WT-3D}
  M(\mu^2,n^2) = - \frac{W(c_0, s_1)}{\Delta(\mu^2,n^2)} = - \frac{D(\mu^2,n^2)}{\Delta(\mu^2,n^2)}, \quad N(\mu^2,n^2) = - \frac{W(c_1, s_0)}{\Delta(\mu^2,n^2)} = - \frac{E(\mu^2,n^2)}{\Delta(\mu^2,n^2)}.
\end{equation}

\begin{rem} \label{M-WellDefined-3D}
For all $n \in \Z$, the WT function $\mu^2 \mapsto M(\mu^2,n^2)$ is meromorphic on $\C$ and has an infinite and discrete set of poles $\{ \alpha^2_{nj}\}_{j \geq 1}$ corresponding to "minus" the Dirichlet eigenvalues of $-\frac{d}{dx^2} + q_{\lambda n}$ or equivalently, corresponding to the zeros of the characteristic function $\mu^2 \mapsto \Delta(\mu^2,n^2)$. Let us show that the integers $m^2, \ m \in \Z$ cannot be poles of $\mu^2 \mapsto M(\mu^2,n^2)$ under our general assumption. Assume the converse, \textit{i.e.} there exists $m \in \Z$ such that $m^2$ is a pole of $M(\mu^2,n^2)$. Thus $-m^2$ is a Dirichlet eigenvalue of $-\frac{d}{dx^2} + q_{\lambda n}$ and we denote by $u_{mn}(x)$ the associated eigenfunction solution of (\ref{Eq3-3D}). Then the function $u(x,y) = u_{mn}(x) Y_{mn}(y,z) $ is a nontrivial solution of (\ref{Eq1}) with Dirichlet boundary conditions. We conclude that such a $u$ is an eigenfunction of $-\Delta_g$ associated to the Dirichlet eigenvalue $\lambda^2$. But this case has been ruled out from the beginning since we assume that $\lambda^2$ is not a Dirichlet eigenvalue of $\Delta_g$. Whence the contradiction.
\end{rem}

We now come back to the expression (\ref{DN-Partiel-2}) of the partial DN map $\Lambda^{mn}_g(\lambda^2)$. For all $m,n \in \Z$, we need to express $v_{mn}'(0)$ and $v_{mn}'(1)$ in terms of $\psi^0_{mn}$ and $\psi^1_{mn}$ in order to find the expressions of the coefficients $L^{mn}(\lambda^2), T^{mn}_R(\lambda^2), T^{mn}_L(\lambda^2), R^{mn}(\lambda^2)$. But the solution $v_{mn}$ of (\ref{Eq3-3D}) can be written as linear combinations of the FSS
$$
  \{ c_0(x,m^2,n^2), s_0(x,m^2,n^2)\}, \quad \quad \{ c_1(x,m^2,n^2), s_1(x,m^2,n^2)\}.
$$	
Precisely, we write
$$
  v_{mn}(x) = \alpha \,c_0(x,m^2,n^2) + \beta \,s_0(x,m^2,n^2) = \gamma \,c_1(x,m^2,n^2) + \delta \,s_1(x,m^2,n^2),
$$
for some constants $\alpha,\beta,\gamma,\delta$. Using (\ref{Eq3-3D}) and (\ref{FSS-3D}), we first get
\begin{equation} \label{a1}
  \left( \begin{array}{c} v_{mn}(0) \\ v_{mn}(1) \end{array} \right) = \left( \begin{array}{c} (h(0))^{1/4} \psi^0_{mn} \\ (h(1))^{1/4} \psi^1_{mn} \end{array} \right) = \left( \begin{array}{c} \alpha \\ \gamma \end{array} \right) = \left( \begin{array}{c} \gamma \,c_1(0,m^2,n^2) + \delta \, s_1(0,m^2,n^2),  \\ \alpha \,c_0(1,m^2,n^2) + \beta \,s_0(1,m^2,n^2) \end{array} \right).
\end{equation}
From this we obtain in particular
\begin{equation} \label{a2}
  \left( \begin{array}{c} \beta \\ \delta  \end{array} \right) = \left( \begin{array}{cc} -(h(0))^{1/4} \frac{c_0(1,m^2,n^2)}{s_0(1,m^2,n^2)} & \frac{(h(1))^{1/4}}{s_0(1,m^2,n^2)} \\ \frac{(h(0))^{1/4}}{s_1(0,m^2,n^2)} & -(h(1))^{1/4}\frac{c_1(0,m^2,n^2)}{s_1(0,m^2,n^2)} \end{array} \right) \left( \begin{array}{c} \psi^0_{mn} \\ \psi^1_{mn}  \end{array} \right).
\end{equation}
Also using (\ref{FSS-3D}) and (\ref{a2}), we have
\begin{equation} \label{a3}
  \left( \begin{array}{c} v_{mn}'(0) \\ v_{mn}'(1)  \end{array} \right) = \left( \begin{array}{c} \beta \\ \delta  \end{array} \right) = \left( \begin{array}{cc} -(h(0))^{1/4} \frac{c_0(1,m^2,n^2)}{s_0(1,m^2,n^2)} & \frac{(h(1))^{1/4}}{s_0(1,m^2,n^2)} \\ \frac{(h(0))^{1/4}}{s_1(0,m^2,n^2)} & -(h(1))^{1/4}\frac{c_1(0,m^2,n^2)}{s_1(0,m^2,n^2)} \end{array} \right) \left( \begin{array}{c} \psi^0_{mn} \\ \psi^1_{mn}  \end{array} \right).
\end{equation}
Therefore by (\ref{DN-Partiel-0}) and (\ref{a1}) - (\ref{a3}), we obtain for all $\psi \in H^{1/2}(T^2)$
$$ 
  \Lambda^{mn}_g(\lambda^2) \left( \begin{array}{c} \psi_{mn}^0 \\ \psi_{mn}^1  \end{array} \right) = \left( \begin{array}{cc} \frac{(\log h)'(0)}{4\sqrt{f(0)}} + \frac{1}{\sqrt{f(0)}} \frac{c_0(1,m^2,n^2)}{s_0(1,m^2,n^2)} &  - \frac{1}{\sqrt{f(0)}} \frac{h^{1/4}(1)}{h^{1/4}(0)} \frac{1}{s_0(1,m^2,n^2)} \\ \frac{1}{\sqrt{f(1)}} \frac{h^{1/4}(0)}{h^{1/4}(1)} \frac{1}{s_1(0,m^2,n^2)} &   -\frac{(\log h)'(1)}{4\sqrt{f(1)}} - \frac{1}{\sqrt{f(1)}} \frac{c_1(0,m^2,n^2)}{s_1(0,m^2,n^2)} \end{array} \right) \left( \begin{array}{c} \psi_{mn}^0 \\ \psi_{mn}^1  \end{array} \right),
$$
or equivalently,
$$ 
  \Lambda^{mn}_g(\lambda^2) = \left( \begin{array}{cc} \frac{(\log h)'(0)}{4\sqrt{f(0)}} + \frac{1}{\sqrt{f(0)}} \frac{c_0(1,m^2,n^2)}{s_0(1,m^2,n^2)} &  - \frac{1}{\sqrt{f(0)}} \frac{h^{1/4}(1)}{h^{1/4}(0)} \frac{1}{s_0(1,m^2,n^2)} \\ \frac{1}{\sqrt{f(1)}} \frac{h^{1/4}(0)}{h^{1/4}(1)} \frac{1}{s_1(0,m^2,n^2)} &   -\frac{(\log h)'(1)}{4\sqrt{f(1)}} - \frac{1}{\sqrt{f(1)}} \frac{c_1(0,m^2,n^2)}{s_1(0,m^2,n^2)} \end{array} \right) .
$$
Finally, using (\ref{FSS-3D}) again, we easily show that
\begin{equation} \label{DN-Partiel-5}
  \Lambda^{mn}_g(\lambda^2) = \left( \begin{array}{cc} \frac{(\log h)'(0)}{4\sqrt{f(0)}} - \frac{1}{\sqrt{f(0)}} M(m^2,n^2) &  - \frac{1}{\sqrt{f(0)}} \frac{h^{1/4}(1)}{h^{1/4}(0)} \frac{1}{\Delta(m^2,n^2)} \\ -\frac{1}{\sqrt{f(1)}} \frac{h^{1/4}(0)}{h^{1/4}(1)} \frac{1}{\Delta(m^2,n^2)} &   -\frac{(\log h)'(1)}{4\sqrt{f(1)}} - \frac{1}{\sqrt{f(1)}} N(m^2,n^2) \end{array} \right) .
\end{equation}

\begin{rem}
1. For all $m,n \in \Z$, the characterictic and WT functions $\Delta(m^2,n^2)$ and $M(m^2,n^2)$ are well defined thanks to our assumption stating that $\lambda^2$ is not a Dirichlet eigenvalue of $\Delta_g$ (see Remark \ref{M-WellDefined-3D} for this point). \\
2. There is of course no importance in the choice of the parameters $m$ or $n$ to be fixed. If $m \in \Z$ was to be fixed, we could define the same objects with respect to the parameter $\nu^2 \in \C$ which would replace the parameter $n^2$ in (\ref{Eq4-3D}). \\
\end{rem}

The result above shows that the coefficients of the partial DN map $\Lambda^{mn}_g(\lambda^2)$ can be simply interpreted in terms of the characteristic and Weyl-Titchmarsh functions associated to the ODE (\ref{Eq4-3D}) - (\ref{BC4-3D}). In particular, the partial DN maps $L(\lambda^2)$ and $R(\lambda^2)$ which correspond to the global DN map restricted to $\Gamma_0$ and $\Gamma_1$ respectively, only depend on the Weyl-Titchmarsh functions $M(m^2, n^2)$ (modulo some explicit constants). In fact, in our model, the knowledge of $L(\lambda^2)$ and $R(\lambda^2)$ is clearly equivalent to the knowledge of the sequence of Weyl-Titchmarsh functions $\{M(m^2, n^2)\}_{m,n \in \Z}$ (modulo some constants).

Similarly, the knowledge of the partial DN maps $T_L(\lambda^2)$ and $T_R(\lambda^2)$ which correspond to the global DN map whose data are measured on the disjoint sets $\Gamma_0$ and $\Gamma_1$, is equivalent to the knowledge of the characteristic functions $\Delta(m^2, n^2)$ (modulo some explicit constants).

We finish this Section with a simple result that is the key point of our non-uniqueness results for the anisotropic Calderon problem with data measured on disjoint sets.

\begin{lemma} \label{Uni-Pot-q}
  Consider two potentials $q_{\lambda n}$ and $\tilde{q}_{\lambda n}$ of the form (\ref{Eq4-3D}) associated to two metrics $g$ and $\tilde{g}$ of the form (\ref{Metric}). Assume that
$$ 	
	q_{\lambda n} = \tilde{q}_{\lambda n},
$$	
for at least two different  $n \in \Z$. Then
\begin{equation} \label{Uni-Pot}
  \tilde{f}(x) = c^4(x) f(x), \quad \quad \tilde{h}(x) = c^4(x) h(x),
\end{equation}
where the positive function $c$ satisfies the non-linear ODE
\begin{equation} \label{ODE-c}
  c'' + \frac{1}{2} (\log h)' c' + \lambda^2 f ( c - c^5) = 0.
\end{equation}
If moreover, $\lambda^2 = 0$, then we can solve (\ref{ODE-c}) explicitly and get
\begin{equation} \label{c}
  c(x) = \left( A + B \int_0^x \frac{ds}{\sqrt{h(s)}} \right),
\end{equation}
where $(A,B)$ are any constants such that $c(x) > 0$ for all $x \in [0,1]$.
\end{lemma}

\begin{proof}
Assume that $q_{\lambda n} = \tilde{q}_{\lambda n}$ for two different $n \in \Z$. Then using the definition (\ref{Eq4-3D}) for these two different values of $n$, we immediately get
\begin{equation} \label{x1}
  \left\{ \begin{array}{rcl} \frac{f}{h} & = & \frac{\tilde{f}}{\tilde{h}}, \\
	\frac{(\log h)'^2}{16} + \frac{(\log h)''}{4} - \lambda^2 f & = & \frac{(\log \tilde{h})'^2}{16} + \frac{(\log \tilde{h})''}{4} - \lambda^2 \tilde{f}. \end{array} \right.
\end{equation}
We set $c = \left( \frac{\tilde{h}}{h} \right)^{1/4}$. Hence a short calculation shows that $c$ must satisfy (\ref{ODE-c}). Moreover we get easily (\ref{Uni-Pot}) from (\ref{x1}). Finally, when $\lambda^2 = 0$, the ODE (\ref{ODE-c}) becomes a first order linear ODE in the unknown $c'$ and its solution leads to (\ref{c}).
\end{proof}


\subsection{The 3D anisotropic Calderon problem with data measured on the same connected component} \label{Calderon-R}

In this Section, we solve the anisotropic Calderon problem in the case where the Dirichlet and Neumann data are measured on non-empty open sets  $\Gamma_D$ and $\Gamma_N$ belonging to the same connected component of $\partial M$. Precisely, we prove

\begin{thm} \label{MainThm-R-3D}
Let $(M,g)$ and $(M,\tilde{g})$ denote two Riemannian manifolds of the form (\ref{Metric}), \textit{i.e.}
$$
	g = f(x) dx^2 + f(x) dy^2 + h(x) dz^2, \quad \tilde{g} = \tilde{f}(x) dx^2 + \tilde{f}(x) dy^2 + \tilde{h}(x) dz^2.
$$
We shall add the subscript $\ \tilde{}$ to all the quantities referring to $(M,\tilde{g})$. Let $\lambda^2 \in \R$. Let $\Gamma_D$ and $\Gamma_N$ be non-empty open sets belonging to the same connected component of $\partial M$. Assume that $\Gamma_D \cap \Gamma_N \ne \emptyset$, with $\Gamma_N$ containing an annular region of the type
$$
  (y_0 - \delta, y_0 + \delta) \times T^1, \quad \textrm{or} \quad T^1 \times (z_0 - \delta, z_0 + \delta),
$$
where $\delta > 0$. Assume moreover that $\Lambda_{g, \Gamma_D,\Gamma_N}(\lambda^2) = \Lambda_{\tilde{g}, \Gamma_D,\Gamma_N}(\lambda^2)$. Then
$$
  \tilde{g} = g.
$$	

\end{thm}

To prove this Theorem, we shall use the Complex Angular Momentum (CAM) method with respect to the parameter $m^2$, that is to say that we shall allow $m^2$ to be complex and use the beautiful analytic properties of the different objects related to the DN map (such as the characteristic function $\Delta(m^2,n^2)$ and the Weyl-Titchmarsh functions $M(m^2,n^2)$ and $N(m^2,n^2)$. In what follows, we shall use freely the results recalled in Section \ref{Preliminary} on the functions $\Delta(\mu^2, n^2), D(\mu^2, n^2), E(\mu^2, n^2), M(\mu^2, n^2), \N(\mu^2, n^2)$ with $\mu \in \C$ and for fixed $n \in \Z$.

\begin{proof}
We assume that $\Gamma_D, \Gamma_N \subset \Gamma_0 = \{0\} \times T^2$. The proof for $\Gamma_D, \Gamma_N \subset \Gamma_1 = \{1\} \times T^2 $ is the same and we omit it. \\

\noindent Assume also (without loss of generality) that $\Gamma_N$ contains the annular open set
$$
  ([-\pi, -\pi + \delta] \times T^1) \cup ( \pi - \delta, \pi] \times T^1),
$$
for a small $\delta > 0$ and that $\Lambda_{g, \Gamma_D,\Gamma_N}(\lambda^2) = \Lambda_{\tilde{g}, \Gamma_D,\Gamma_N}(\lambda^2)$. According to (\ref{DN-Partiel-5}) and the discussion after it, this assumption is equivalent to
\begin{equation} \label{b1}
  \Phi(y,z) := \sum_{m,n} \left[ A - \frac{M(m^2,n^2)}{\sqrt{f(0)}} - \tilde{A} + \frac{\tilde{M}(m^2,n^2)}{\sqrt{\tilde{f}(0)}} \right] \hat{\psi}(m,n) e^{imy + inz} = 0,
\end{equation}
for all $(y,z) \in \Gamma_N$ and for all $\psi \in H^{1/2}(T^2)$ with supp $\,\psi \subset \Gamma_D$. Here $\hat{\psi}(m,n)$ denotes the Fourier coefficients of $\psi$ and
$$
  A = \frac{(\log h)'(0)}{4\sqrt{f(0)}}, \quad \tilde{A} = \frac{(\log \tilde{h})'(0)}{4\sqrt{\tilde{f}(0)}}.
$$

Let us assume from now on that $\psi$ is smooth enough (say $\psi \in H^1(T^1)$) and that supp $\,\psi \subset \Gamma_D \cap [- r, r] \times T^1$ where $r < \pi$. This is always possible since $\Gamma_D$ is open in $T^2$ and up to consider $\psi$ with smaller support than $\Gamma_D$. Then we can extract several informations from (\ref{b1}).

First the functions $\Phi$ and $\psi$ are in $L^2(T^2)$. Second, supp $\,\Phi \subset [-R,R] \times T^1$ with $R < \pi$ and supp $\,\psi \subset [- r, r] \times T^1$ where $r < \pi$. We thus conclude from the multivariable Paley-Wiener Theorem (see \cite{Ho1}, Theorem 7.3.1, p. 181) that the Fourier transforms of $\Phi$ and $\psi$ are entire functions on $\C^2$ that satisfy the estimates
\begin{equation} \label{b2}
  \left\{ \begin{array}{c} |\hat{\psi}(\mu,\nu)| \leq Ce^{r |\Im(\mu)| + \pi |\Im(\nu)|}, \\
  |\hat{\Phi}(\mu,\nu)| \leq Ce^{R |\Im(\mu)| + \pi |\Im(\nu)|}, \end{array} \right. \quad \forall (\mu,\nu) \in \C^2.
\end{equation}
Third, we can check directly that for all $m,n \in \Z$
\begin{equation} \label{b3}
  \hat{\Phi}(m,n) = \left[ A - \frac{M(m^2,n^2)}{\sqrt{f(0)}} - \tilde{A} + \frac{\tilde{M}(m^2,n^2)}{\sqrt{\tilde{f}(0)}} \right] \hat{\psi}(m,n).
\end{equation}
Recalling the definition (\ref{WT-3D}) of the Wey-Titchmarsh functions $M$ and $\tilde{M}$, this latter equality can be rewritten as
\begin{equation} \label{b4}
  \Delta(m^2,n^2) \tilde{\Delta}(m^2,n^2) \hat{\Phi}(m,n) - \left[ A - \frac{D(m^2,n^2) \tilde{\Delta}(m^2,n^2)}{\sqrt{f(0)}}  - \tilde{A} + \frac{\tilde{D}(m^2,n^2) \Delta(m^2,n^2)}{\sqrt{\tilde{f}(0)}} \right] \hat{\psi}(m,n)  = 0,
\end{equation}
for all $m,n \in \Z$.

Let us fix $n \in \Z$. We denote by $F(\mu,n)$ the function in the left-hand-side of (\ref{b4}) where $m^2$ is replaced by $\mu$. Clearly, $F$ is entire of order $1$  with respect to $\mu$ and satisfies the estimate on the imaginary axis:
$$
  |F(iy,n)| \leq C e^{\max(r,R) |y|}, \quad \forall y \in \R,
$$
thanks to Corollary \ref{OrderHalf} and (\ref{b2}). Moreover, $F$ vanishes on the integers by (\ref{b3}). Since $\max(r,R) < \pi$, we deduce from Carlson's Theorem (see \cite{Boa}, Theorem 9.2.1) that $F(\mu,n) = 0$ for all $\mu \in \C, \ n \in \Z$. This can be written equivalently as
\begin{equation} \label{b5}
  \hat{\Phi}(\mu,n) = \left[ A - \frac{M(\mu^2,n^2)}{\sqrt{f(0)}} - \tilde{A} + \frac{\tilde{M}(\mu^2,n^2)}{\sqrt{\tilde{f}(0)}} \right] \hat{\psi}(\mu,n), \quad \forall \mu \in \C, \ n \in \Z.
\end{equation}
Since the function $\hat{\Phi}(\mu,n)$ is entire w.r.t. $\mu$, the function $\left[ A - \frac{M(\mu^2,n^2)}{\sqrt{f(0)}} - \tilde{A} + \frac{\tilde{M}(\mu^2,n^2)}{\sqrt{\tilde{f}(0)}} \right] \hat{\psi}(\mu,n)$ must also be entire in $\mu$ for all $\psi$ smooth enough such that supp $\,\psi \subset \Gamma_D \cap [- r, r] \times T^1$ where $r < \pi$. In other worlds, the poles of $\frac{M(\mu^2,n^2)}{\sqrt{f(0)}}$ and $\frac{\tilde{M}(\mu^2,n^2)}{\sqrt{\tilde{f}(0)}}$ as well their residues must coincide. In fact, an easy adaption then of the argument given in the proof of Proposition \ref{equalityWT} shows that (\ref{b5}) implies

\begin{equation}\label{b6}
\Delta (\mu^2,n^2) = \tilde{\Delta}(\mu^2,n^2), \quad \forall \mu \in \C, \ n \in \Z.
\end{equation}
\begin{equation} \label{b7}
M (\mu^2, n^2)- M(0,n^2) = \tilde{M}(\mu^2,n^2) - \tilde{M}(0,n^2), \quad \forall \mu \in \C \setminus \R, \ n \in \Z,
\end{equation}
and
\begin{equation} \label{b8}
f(0) = \tilde{f}(0).
\end{equation}

We can now finish the proof as follows. From (\ref{b7}) and the Borg-Marchenko Theorem \ref{BM}, we deduce first that
\begin{equation} \label{b9}
  q_{\lambda n}(x) = \tilde{q}_{\lambda n}(x), \quad \forall x \in [0,1], \ n \in \Z,
\end{equation}
which the main assumption of Lemma \ref{Uni-Pot-q}. In particular, we know that $\tilde{f} = c^4 f$ and $\tilde{h} = c^4 h$ where $c$ is a function satisfying the ODE (\ref{ODE-c}).

But the equality (\ref{b9}) in turn implies that $M (\mu^2, n^2) = \tilde{M}(\mu^2,n^2)$ for all $\mu \in \C, \ n \in \Z$. Putting this in (\ref{b1}), we get
\begin{equation}\label{b10}
  (A - \tilde{A}) \psi(y,z) = 0, \quad \forall (y,z) \in \Gamma_N,
\end{equation}
and for all $\psi$ such that supp $\,\psi \subset \Gamma_D$. Since $\Gamma_D \cap \Gamma_N \ne \emptyset$, we conclude from (\ref{b10}) that
\begin{equation}\label{b11}
  A = \tilde{A}.
\end{equation}
Together with (\ref{b8}), (\ref{b11}) implies that the function $c = \left( \frac{\tilde{h}}{h} \right)^{1/4}$ satisfies the Cauchy conditions
\begin{equation}\label{b12}
  c(0) = 1, \quad c'(0) = 0.
\end{equation}
Hence, we deduce from (\ref{b12}) that $c = 1$ is the unique solution of the ODE (\ref{ODE-c}). In other words, we have proved that $ \tilde{f} = f$ and $\tilde{h} = h$ or equivalently, $\tilde{g} = g$ which concludes the proof of the Theorem.
\end{proof}

We now state our first non-uniqueness result for  three dimensional  Riemannian manifolds.

\begin{thm} \label{MainThm-NonUniqueness-R-3D}
Let $(M,g)$ and $(M,\tilde{g})$ denote Riemannian manifolds of the form (\ref{Metric}), \textit{i.e.}
$$
	g = f(x) dx^2 + f(x) dy^2 + h(x) dz^2, \quad \tilde{g} = \tilde{f}(x) dx^2 + \tilde{f}(x) dy^2 + \tilde{h}(x) dz^2.
$$
Let $\lambda^2 \in \R$. Let $\Gamma_D$ and $\Gamma_N$ be non-empty open sets belonging to the same connected component of $\partial M$. Assume that $\Gamma_D \cap \Gamma_N = \emptyset$. Then there exists infinitely many pairs of non-isometric metrics $(g, \tilde{g})$ given by $\tilde{g} = c^4 g$ where $c$ are smooth positive strictly increasing or decreasing functions such that $c(0) = 1$ if $\Gamma_D, \Gamma_N \subset \Gamma_0$ and $c(1) = 1$ if $\Gamma_D, \Gamma_N \subset \Gamma_1$, satisfying
$$
  \Lambda_{g, \Gamma_D,\Gamma_N}(\lambda^2) = \Lambda_{\tilde{g}, \Gamma_D,\Gamma_N}(\lambda^2).
$$		
\end{thm}

\begin{proof}
Let us assume that $\Gamma_D, \Gamma_N \subset \Gamma_0 = \{0\} \times T^1$. We construct pairs of metrics $(g, \tilde{g})$ satisfying $\Lambda_{g, \Gamma_D,\Gamma_N}(\lambda^2) = \Lambda_{\tilde{g}, \Gamma_D,\Gamma_N}(\lambda^2)$ as follows. Let $f$ and $c$ be any smooth positive function on $[0,1]$ such that $c(0) = 1$ and $c'(x) \ne 0$ for all $x \in [0,1]$. In other words, $c(x)$ is a strictly monotonic function. Define
\begin{equation} \label{b13}
  h = C e^{-2 \int_0^x \frac{c''(s) + \lambda^2 f(s) (c(s) - c^5(s))}{c'(s)} ds}, \quad \tilde{f} = c^4 f, \quad \tilde{h} = c^4 h.
\end{equation}
Clearly, we have then $\tilde{g} = c^4 g$. Using Lemma \ref{Uni-Pot-q}, it is immediate to check that $q_{\lambda n} = \tilde{q}_{\lambda n}$ for all $n \in \Z$ where $q_{\lambda n}$ and $\tilde{q}_{\lambda n}$ are given by (\ref{Eq4-3D}). In particular, for such choices of metrics $(g,\tilde{g})$, we always have $M(m^2,n^2) = \tilde{M}(m^2,n^2)$ for all $m,n\in \Z$. Moreover, our assumption $c(0) = 1$ implies that $f(0) = \tilde{f}(0)$.

Now, using (\ref{DN-Partiel-5}) and the discussion after it, we see that for all $\psi \in H^{1/2}(T^2)$ with supp $\,\psi \subset \Gamma_D$, we have
\begin{eqnarray}
  \Lambda_{g, \Gamma_D,\Gamma_N}(\lambda^2)(\psi) & :=  \left[ \ds\sum_{m,n} \left( \frac{(\log h)'(0)}{4\sqrt{f(0)}} - \frac{M(m^2,n^2)}{\sqrt{f(0)}} \right) \hat{\psi}(m,n) e^{imy + inz} \right]_{| \,(y,z) \in \Gamma_N}, \\
	         & = \left[ \frac{(\log h)'(0)}{4\sqrt{f(0)}} \psi(y,z) -  \ds\sum_{m,n} \frac{M(m^2,n^2)}{\sqrt{f(0)}} \hat{\psi}(m,n) e^{imy + inz} \right]_{| \,(y,z) \in \Gamma_N}.
\end{eqnarray}
But since $\Gamma_D \cap \Gamma_N = \emptyset$, we thus deduce that
$$
  \frac{(\log h)'(0)}{4\sqrt{f(0)}} \psi(y,z) = 0,
$$
for all $(y,z) \in \Gamma_N$. Hence we obtain
\begin{eqnarray}
  \Lambda_{g, \Gamma_D,\Gamma_N}(\lambda^2)(\psi) & = - \left[ \ds\sum_{m,n} \frac{M(m^2,n^2)}{\sqrt{f(0)}} \hat{\psi}(m,n) e^{imy + inz} \right]_{| \,(y,z) \in \Gamma_N}, \\
	     & = - \left[ \ds\sum_{m,n} \frac{\tilde{M}(m^2,n^2)}{\sqrt{\tilde{f}(0)}} \hat{\psi}(m,n) e^{imy + inz} \right]_{| \,(y,z) \in \Gamma_N}, \\
			 & = \Lambda_{\tilde{g}, \Gamma_D,\Gamma_N}(\lambda^2)(\psi), \hspace{4.8cm}
\end{eqnarray}
for all $\psi \in H^{1/2}(T^2)$ with supp $\,\psi \subset \Gamma_D$. This proves the result.
\end{proof}


\begin{rem}
  If we fix a Riemannian manifold $(M,g)$ of the form (\ref{Metric}), we don't know a priori whether there exists a metric $\tilde{g}$ such that $\Lambda_{g, \Gamma_D,\Gamma_N}(\lambda^2) = \Lambda_{\tilde{g}, \Gamma_D,\Gamma_N}(\lambda^2)$ with $\Gamma_D \cap \Gamma_N = \emptyset$ except if the original metric $g$ has the form
$$
	g = f(x) dx^2 + f(x) dy^2 + h(x) dz^2,
$$
where $h$ given by (\ref{b13}). To be able to treat the general case, that is the case in which $f,h$ are any smooth positive functions on $[0,1]$, would require to know whether the non-linear ODE
\begin{equation} \label{b14}
  c'' + \frac{1}{2} (\log h)' c' + \lambda^2 f ( c - c^5) = 0,
\end{equation}
has global solutions on $[0,1]$ satisfying $c(0) = 1$. It is not difficult to see that this is not the case: we can easily construct a solution $c(x)$ of (\ref{b14}) with $h \equiv 1$ with the asymptotics $c(x) \sim (1-x)^{-\half}$ when $x \rightarrow 1^-$.
\end{rem}

In the case of zero frequency $\lambda^2 = 0$, we can do a little better since we can solve explicitly the ODE (\ref{b14}). Since this result is interesting in its own sake, we state it as a Theorem

\begin{thm} \label{MainThm-NonUniquenessZeroFrequency-R-3D}
Let $(M,g)$ denotes a Riemannian manifold of the form (\ref{Metric}), \textit{i.e.}
$$
	g = f(x) dx^2 + f(x) dy^2 + h(x) dz^2.
$$
Let $\Gamma_D$ and $\Gamma_N$ be non-empty open sets belonging to the same connected component of $\partial M$. Assume that $\Gamma_D \cap \Gamma_N = \emptyset$. Then \\

\noindent 1) if $\Gamma_D, \Gamma_N \subset \Gamma_0$, there exists a one parameter family of metrics $\tilde{g}$ given by
$$
\tilde{g} = \left[  1 +  B \int_0^x \frac{ds}{\sqrt{h(s)}} \right]^4 g, \quad B > 0,
$$
that satisfies $\Lambda_{g, \Gamma_D,\Gamma_N}(0) = \Lambda_{\tilde{g}, \Gamma_D,\Gamma_N}(0)$.	\\

\noindent 2) if $\Gamma_D, \Gamma_N \subset \Gamma_1$, there exists a one parameter family of metrics $\tilde{g}$ given by
$$
\tilde{g} = \left[  1 +  B \int_x^1 \frac{ds}{\sqrt{h(s)}} \right]^4 g, \quad B > 0,
$$
that satisfies $\Lambda_{g, \Gamma_D,\Gamma_N}(0) = \Lambda_{\tilde{g}, \Gamma_D,\Gamma_N}(0)$.	\\
\end{thm}

\begin{proof}
  We only prove 1. Assume that $\lambda^2 = 0$. For arbitrary smooth positive functions $f,h$ on $[0,1]$, the function $c$ solution of (\ref{b14}) is explicitly given by
$$
  c(x) = A + B \int_0^x \frac{ds}{\sqrt{h(s)}},
$$	
for some constants $A,B$. Since we also demand $c(0) =  1$ and $c^4 > 0$, we only consider the one parameter solutions
$$
  c(x) = 1 +  B \int_0^x \frac{ds}{\sqrt{h(s)}},
$$ 	
for some constant $B > 0$. Then we set $\tilde{g} = c^4 g$ and we use the same proof as in Theorem \ref{MainThm-NonUniqueness-R-3D} to conclude that $\Lambda_{g, \Gamma_D,\Gamma_N}(0) = \Lambda_{\tilde{g}, \Gamma_D,\Gamma_N}(0)$.
\end{proof}

\subsection{The 3D anisotropic Calderon problem with data measured on distinct connected components} \label{Calderon-T}

In this Section, we prove our third main Theorem, namely we give a counterexample to uniqueness for the anisotropic Calderon problem for three dimensional Riemannian manifolds with data measured on two different connected components of the boundary.

\begin{thm} \label{MainThm-T-3D}
Let $(M,g)$ and $(M,\tilde{g})$ denote Riemannian manifolds of the form (\ref{Metric}), \textit{i.e.}
$$
	g = f(x) dx^2 + f(x) dy^2 + h(x) dz^2, \quad \tilde{g} = \tilde{f}(x) dx^2 + \tilde{f}(x) dy^2 + \tilde{h}(x) dz^2.
$$
Let $\lambda^2 \in \R$. Let $\Gamma_D$ and $\Gamma_N$ be non-empty open sets belonging to two different connected components of $\partial M$. Then there exist infinitely many pairs of non-isometric metrics $(g, \tilde{g})$ given by $\tilde{g} = c^4 g$ where $c$ are smooth positive strictly increasing or decreasing functions such that $c(1)^3 = c(0) \ne 1$ if $\Gamma_D \subset \Gamma_0$ and $\Gamma_N \subset \Gamma_1$ or $c(0)^3 = c(1) \ne 1$ if $\Gamma_D \subset \Gamma_1$ and $\Gamma_N \subset \Gamma_0$, satisfying
$$
  \Lambda_{g, \Gamma_D,\Gamma_N}(\lambda^2) = \Lambda_{\tilde{g}, \Gamma_D,\Gamma_N}(\lambda^2).
$$		
\end{thm}

\begin{proof}
The proof is almost the same as that of Theorem \ref{MainThm-NonUniqueness-R-3D}.

Let us assume that $\Gamma_D \subset \Gamma_0 = \{0\} \times T^1$ and $\Gamma_D \subset \Gamma_1 = \{1\} \times T^1$. We construct pairs of metrics $(g, \tilde{g})$ satisfying $\Lambda_{g, \Gamma_D,\Gamma_N}(\lambda^2) = \Lambda_{\tilde{g}, \Gamma_D,\Gamma_N}(\lambda^2)$ as follows. Let $f$ and $c$ be any smooth positive function on $[0,1]$ such that $c(1)^3 = c(0)$ and $c'(x) \ne 0$ for all $x \in [0,1]$.

Define
\begin{equation} \label{j1}
  h = C e^{-2 \int_0^x \frac{c''(s) + \lambda^2 f(s) (c(s) - c^5(s))}{c'(s)} ds}, \quad \tilde{f} = c^4 f, \quad \tilde{h} = c^4 h.
\end{equation}
Clearly, we have then $\tilde{g} = c^4 g$. Using Lemma \ref{Uni-Pot-q}, it is immediate to check that $q_{\lambda n} = \tilde{q}_{\lambda n}$ for all $n \in \Z$ where $q_{\lambda n}$ and $\tilde{q}_{\lambda n}$ are given by (\ref{Eq4-3D}). In particular, for such choices of metrics $(g,\tilde{g})$, we always have $\Delta(m^2,n^2) = \tilde{\Delta}(m^2,n^2)$ for all $m,n\in \Z$. Moreover, our assumption $c(1)^3 = c(0)$ is then equivalent to
$$
  \frac{1}{\sqrt{f(1)}} \left( \frac{h(0)}{h(1)} \right)^{1/4} = \frac{1}{\sqrt{\tilde{f}(1)}} \left( \frac{\tilde{h}(0)}{\tilde{h}(1)} \right)^{1/4}.
$$	

Now, using (\ref{DN-Partiel-5}) and the ensuing discussion, we see that for all $\psi \in H^{1/2}(T^2)$ with supp $\,\psi \subset \Gamma_D$, we have
\begin{eqnarray}
  \Lambda_{g, \Gamma_D,\Gamma_N}(\lambda^2)(\psi) & :=  - \left[ \ds\sum_{m,n} \left( \frac{1}{\sqrt{f(1)}} \left( \frac{h(0)}{h(1)} \right)^{1/4} \frac{1}{\Delta(m^2,n^2)} \right) \hat{\psi}(m,n) e^{imy + inz} \right]_{| \,(y,z) \in \Gamma_N}, \\
	        & =  - \left[ \ds\sum_{m,n} \left( \frac{1}{\sqrt{\tilde{f}(1)}} \left( \frac{\tilde{h}(0)}{\tilde{h}(1)} \right)^{1/4} \frac{1}{\tilde{\Delta}(m^2,n^2)} \right) \hat{\psi}(m,n) e^{imy + inz} \right]_{| \,(y,z) \in \Gamma_N}, \\
					& = \Lambda_{\tilde{g}, \Gamma_D,\Gamma_N}(\lambda^2)(\psi). \hspace{8.7cm}	
\end{eqnarray}
Since the above equality holds for all $\psi \in H^{1/2}(T^2)$ with supp $\,\psi \subset \Gamma_D$, the result is proved.  	
\end{proof}

Finally, we can slightly improve our non-uniqueness result in the case of zero frequency. We have

\begin{thm} \label{MainThm-NonUniquenessZeroFrequency-T-3D}
Let $(M,g)$ denotes a Riemannian manifold of the form (\ref{Metric}), \textit{i.e.}
$$
	g = f(x) dx^2 + f(x) dy^2 + h(x) dz^2.
$$
Let $\Gamma_D$ and $\Gamma_N$ be non-empty open sets belonging to two different connected component of $\partial M$. Then \\

\noindent 1) if $\Gamma_D \subset \Gamma_0$ and $\Gamma_N \subset \Gamma_1$, there exists a one parameter family of metrics $\tilde{g}$ given by
$$
\tilde{g} = \left[ A + \frac{A^3 - A}{\int_0^1 \frac{ds}{\sqrt{h(s)}}} \int_x^1 \frac{ds}{\sqrt{h(s)}} \right]^4 g, \quad A > 0,
$$
that satisfies $\Lambda_{g, \Gamma_D,\Gamma_N}(0) = \Lambda_{\tilde{g}, \Gamma_D,\Gamma_N}(0)$.	\\

\noindent 2) if $\Gamma_D \subset \Gamma_1$ and $\Gamma_N \subset \Gamma_0$, there exists a one parameter family of metrics $\tilde{g}$ given by
$$
\tilde{g} = \left[ A + \frac{A^3 - A}{\int_0^1 \frac{ds}{\sqrt{h(s)}}} \int_0^x \frac{ds}{\sqrt{h(s)}} \right]^4 g, \quad A > 0,
$$
that satisfies $\Lambda_{g, \Gamma_D,\Gamma_N}(0) = \Lambda_{\tilde{g}, \Gamma_D,\Gamma_N}(0)$.	\\
\end{thm}

\begin{proof}
  We only prove 1). Assume that $\lambda^2 = 0$. For arbitrary smooth positive functions $f,h$ on $[0,1]$, the function $c$ solution of (\ref{ODE-c}) is explicitly given by
$$
  c(x) = A + B \int_x^1 \frac{ds}{\sqrt{h(s)}},
$$	
for some constants $A,B$. Since we also requires that $c(1)^3 = c(0)$, we only consider the one parameter family of functions
$$
  c(x) = A + \frac{A^3 - A}{\int_0^1 \frac{ds}{\sqrt{h(s)}}} \int_x^1 \frac{ds}{\sqrt{h(s)}}
$$ 	
for some constant $A > 0$. Then we set $\tilde{g} = c^4 g$ and we use the same proof as in Theorem \ref{MainThm-NonUniqueness-R-3D} to conclude that $\Lambda_{g, \Gamma_D,\Gamma_N}(0) = \Lambda_{\tilde{g}, \Gamma_D,\Gamma_N}(0)$.
\end{proof}


\subsection{The 3D anisotropic Calderon problem with a potential}

In this Section, we treat the anisotropic Calderon problem \textbf{(Q3)} with a potential $V=V(x) \in L^\infty(M)$, that is a potential depending only on the variable $x$, for our family of metrics (\ref{Metric}). The global DN map $\lambda_{g,V}(\lambda^2)$ associated to the Dirichlet problem
\begin{equation} \label{Schrodinger-3D}
  \left\{ \begin{array}{cc} (-\Delta_g + V) u = \lambda^2 u, & \textrm{on} \ M, \\ u = \psi, & \textrm{on} \ \partial M, \end{array} \right.
\end{equation}
with $\lambda^2$ not belonging to the Dirichlet spectrum of $-\Delta_g + V$ can be constructed in the same way as in Section \ref{3D} since $V = V(x)$ respects the symmetry of $(M,g)$. On each Fourier mode $Y_{mn} = e^{imy + inz}$, we get the following expression for the induced DN map:
\begin{equation} \label{DN-Partiel-Potentiel-3D}
  \Lambda^{mn}_{g,V}(\lambda^2) = \left( \begin{array}{cc} \frac{(\log h)'(0)}{4\sqrt{f(0)}} - \frac{1}{\sqrt{f(0)}} M_V(m^2,n^2) &  - \frac{1}{\sqrt{f(0)}} \frac{h^{1/4}(1)}{h^{1/4}(0)} \frac{1}{\Delta_V(m^2,n^2)} \\ -\frac{1}{\sqrt{f(1)}} \frac{h^{1/4}(0)}{h^{1/4}(1)} \frac{1}{\Delta_V(m^2,n^2)} &   -\frac{(\log h)'(1)}{4\sqrt{f(1)}} - \frac{1}{\sqrt{f(1)}} N_V(m^2,n^2) \end{array} \right),
\end{equation}
where the characteristic and Weyl-Titchmarsh functions $\Delta_V(m^2,n^2), M_V(m^2,n^2)$ and $N_V(m^2,n^2)$ defined by (\ref{Char-3D})-(\ref{WT-3D}) are associated to the radial ODE
\begin{equation} \label{Eq-Schrodinger-3D}
  -v'' + q_{\lambda,V,n}(x) v = - \mu^2 v, \quad \quad q_{\lambda,V,n} = \frac{[(\log h)']^2}{16}  + \frac{(\log h)''}{4}  + n^2 \frac{f}{h} + (V - \lambda^2) f,
\end{equation}
with boundary conditions
\begin{equation} \label{BC-Schrodinger-3D}
  v(0) = 0, \quad v(1) = 0.
\end{equation}
We also recall the dictionary between the above coefficients of the DN map and the notations used in the Introduction
$$
  L_{V}^{mn}(\lambda^2) = \frac{(\log h)'(0)}{4\sqrt{f(0)}} - \frac{1}{\sqrt{f(0)}} M_V(m^2,n^2) = \Lambda^{mn}_{g,V,\Gamma_0}(\lambda^2),
$$
$$
  R_V^{mn}(\lambda^2) = -\frac{(\log h)'(1)}{4\sqrt{f(1)}} - \frac{1}{\sqrt{f(1)}} N_V(m^2,n^2) = \Lambda^{mn}_{g,V,\Gamma_1}(\lambda^2),
$$
$$	
	T_{L,V}^{mn}(\lambda^2) = -\frac{1}{\sqrt{f(1)}} \frac{h^{1/4}(0)}{h^{1/4}(1)} \frac{1}{\Delta_V(m^2,n^2)} = \Lambda^{mn}_{g,V,\Gamma_0, \Gamma_1}(\lambda^2),
$$
$$
  T_{R,V}^{mn}(\lambda^2) = - \frac{1}{\sqrt{f(0)}} \frac{h^{1/4}(1)}{h^{1/4}(0)} \frac{1}{\Delta_V(m^2,n^2)} = \Lambda^{mn}_{g,V,\Gamma_1, \Gamma_0}(\lambda^2),
$$
where $\Gamma_0 = \{0\} \times T^1$ and $\Gamma_1 = \{1\} \times T^1$.

We prove:

\begin{thm} \label{MainThm-Schrodinger-3D}
Let $(M,g)$ a smooth compact Riemannian manifold of the form (\ref{Metric}), \textit{i.e.}
$$
	g = f(x) dx^2 + f(x) dy^2 + h(x) dz^2.
$$
Let $V, \tilde{V} \in L^\infty(M)$ be two potentials that only depend on the variable $x$. Let the frequency $\lambda^2$ be fixed and not belonging to the Dirichlet spectra of $-\Delta_g + V$ and $-\Delta_g + \tilde{V}$. Let $\Gamma_D$, $\Gamma_N$ be nonempty open subsets belonging to the same connected component of $\partial M$, with $\Gamma_N$ containing an annular region of the type
$$
  (y_0 - \delta, y_0 + \delta) \times T^1, \quad \textrm{or} \quad T^1 \times (z_0 - \delta, z_0 + \delta),
$$
where $\delta > 0$. Assume moreover
$$
  \Lambda_{g, V, \Gamma_D, \Gamma_N}(\lambda^2) = \Lambda_{g, \tilde{V}, \Gamma_D, \Gamma_N}(\lambda^2).
$$
Then
$$
	\tilde{V} = V.
$$
\end{thm}

\begin{proof}
Assume for instance that $\Gamma_D, \Gamma_N \subset \Gamma_0$ and $\Lambda_{g, V, \Gamma_D, \Gamma_N}(\lambda^2) = \Lambda_{g, \tilde{V}, \Gamma_D, \Gamma_N}(\lambda^2)$. Then the same proof as in Theorem \ref{MainThm-R-3D} shows first that
\begin{equation} \label{WT=3D}
  M_{V}(\mu^2,n^2) - M_{V}(0,n^2)= M_{\tilde{V}}(\mu^2,n^2)-M_{\tilde{V}}(0,n^2), \quad \forall \mu \in \C \setminus \R, \ \forall n \in \Z.
\end{equation}
Hence the Borg-Marchenko Theorem \ref{BM} gives
\begin{equation} \label{QV=3D}
  q_{\lambda,V,n} = q_{\lambda, \tilde{V},n}, \quad \textrm{on} \ [0,1], \quad \forall n \in \Z.
\end{equation}
We conclude using (\ref{QV=3D}) and (\ref{Eq-Schrodinger-3D}) that $V = \tilde{V}$ on $[0,1]$.
\end{proof}

In the case where the Dirichlet and Neumann data $\Gamma_D$ and $\Gamma_N$ do not belong to the same connected component of $\partial M$, we are able to give very simple counterexamples to uniqueness, in the case where $f=h$. Precisely, we prove:

\begin{thm} \label{SecondThm-Schrodinger-3D}
Let $(M,g)$ a smooth compact Riemannian manifold of the form
$$
	g = f(x) [ dx^2 +  dy^2 +  dz^2 ].
$$
Let $V, \tilde{V} \in L^\infty(M)$ be two potentials that only depend on the variable $x$. Let the frequency $\lambda^2$ be fixed and not belonging to the Dirichlet spectra of $-\Delta_g + V$ and $-\Delta_g + \tilde{V}$. Let $\Gamma_D$, $\Gamma_N$ be nonempty open subsets belonging to distinct connected components of $\partial M$. Then, there exists an infinite dimensional family of explicit potentials $\tilde{V} \in L^\infty([0,1])$ that satisfy
$$
  \Lambda_{g, V, \Gamma_D, \Gamma_N}(\lambda^2) = \Lambda_{g, \tilde{V}, \Gamma_D, \Gamma_N}(\lambda^2).
$$
\end{thm}

\begin{proof}
In the case where $f=h$, the potentials $q_{\lambda,V,n}$ and $q_{\lambda,\tilde{V},n}$ have the following simple form
$$
  q_{\lambda,V,n} = \frac{[(\log f)']^2}{16}  + \frac{(\log f)''}{4} + n^2 + (V(x) - \lambda^2) f(x),
$$
$$
	q_{\lambda,\tilde{V},n} = \frac{[(\log f)']^2}{16}  + \frac{(\log f)''}{4} + n^2 + (\tilde{V}(x) - \lambda^2) f(x).
$$	
Clearly, since $n^2$ does not depend on $x$, the potentials $q_{\lambda,V,n}$, $q_{\lambda,\tilde{V},n}$ are isospectral if and only if the potentials  $\frac{[(\log f)']^2}{16}  + \frac{(\log f)''}{4} + (V(x)-\lambda^2) f(x)$ and $\frac{[(\log f)']^2}{16}  + \frac{(\log f)''}{4} + (\tilde{V}(x)-\lambda^2) f(x)$ are also isospectral. We recall that this is also equivalent (see the proofs of Theorems \ref{MainThm-Schrodinger-2D} and \ref{MainThm-R-3D}) to
\begin{equation} \label{Char=2D-V}
  \Delta_{V}(\mu^2, n^2) = \Delta_{\tilde{V}}(\mu^2,n^2), \quad \forall \mu \in \C, \ \forall n\in \Z,
\end{equation}
which would imply that
\begin{equation} \label{z1}
  \Lambda_{g, V, \Gamma_D, \Gamma_N}(\lambda^2) = \Lambda_{g, \tilde{V}, \Gamma_D, \Gamma_N}(\lambda^2),
\end{equation}
according to (\ref{DN-Partiel-Potentiel-3D}).

We deduce thus from \cite{PT} that given a potential $V \in L^\infty(M)$ as above, there exists an infinite dimensional family of explicit potentials $\tilde{V}$ satisfying (\ref{Char=2D-V}) and thus (\ref{z1}). More precisely, the family
\begin{equation} \label{Iso-Schrodinger-2D-V}
  \tilde{V}_{\lambda,k,t}(x) = V(x) - \frac{2}{f(x)} \frac{d^2}{dx^2} \log \theta_{k,t}(x), \quad \forall k \geq 1, \quad t \in \R,
\end{equation} 	
with
\begin{equation} \label{Iso-Schrodinger1-2D-V}
  \theta_{k,t}(x) = 1 + (e^t - 1) \int_x^1 v_k^2(s) ds,
\end{equation} 		
where $v_k$ is the normalized eigenfunction of (\ref{Eq-Schrodinger-3D}) - (\ref{BC-Schrodinger-3D}) associated to the $k$th-eigenvalue $\alpha^2_k$ satisfies
$$
  \Lambda_{g, V, \Gamma_D, \Gamma_N}(\lambda^2) = \Lambda_{g, \tilde{V}_{\lambda,k,t}, \Gamma_D, \Gamma_N}(\lambda^2), \quad \forall k \geq 1, \quad t \in \R.
$$	
\end{proof}


\end{document}